\documentclass[a4paper, USenglish]{lipics-v2021}

\usepackage{amsthm}
\usepackage{amsmath}
\usepackage{amssymb}
\usepackage{array}
\usepackage{bm}
\usepackage{enumitem}
\usepackage{hyperref}
\usepackage{graphicx}
\usepackage{makecell}
\usepackage{multirow}
\usepackage{xcolor}

\newcommand{\deleted}[1]{}
\newcommand{\rephrase}[3]{\noindent\textbf{#1~#2.}~{\emph{#3}}}
\newcommand{\pa}[1]{\left(#1\right)}

\newcommand{\FT}{\mbox{FT-APX}}		
\newcommand{\Sel}{\mathrm{Select}}	
\newcommand{\SM}{\mbox{ST-Median}}	
\newcommand{\MV}{\mbox{c}} 			
\newcommand{\mb}{\frac{1}{6}}		

\newcommand{\A}{\mathcal{A}}
\newcommand{\tA}{\tilde{\mathcal{A}}}
\DeclareMathOperator{\TD}{\mathfrak{D}}		
\newcommand{\tT}{\tilde{T}}					
\DeclareMathOperator{\LC}{\mathfrak{\eta}}  

\DeclareMathOperator{\Hyper}{Hypergeom}
\newcommand{\hM}{M}				
\newcommand{\hK}{K}				
\newcommand{\m}{m}					
\DeclareMathOperator{\HH}{H}				
\DeclareMathOperator{\PP}{P}				
\DeclareMathOperator{\QQ}{Q}                

\newcommand{\Div}[2]{\operatorname{D} \pa{{#1} \middle\| {#2}}}
\newcommand{\bra}[1]{\left[#1\right]}
\newcommand{\ceil}[1]{\left\lceil#1\right\rceil}

\newcommand{\de}[1]{\mathrm{d}#1}
\newcommand{\Prpr}[1]{\operatorname{Pr}\bra{#1}}

\nolinenumbers

\title{Approximate Selection with Unreliable Comparisons in Optimal Expected Time}

\author{Shengyu Huang}{Department of Computer Science, ETH Z\"{u}rich, Z\"{u}rich, Switzerland}{kumom.huang@gmail.com}{}{}
\author{Chih-Hung Liu}{Department of Computer Science, ETH Z\"{u}rich, Z\"{u}rich, Switzerland}{chih-hung.liu@inf.ethz.ch}{}{}
\author{Daniel Rutschman}{Department of Computer Science, ETH Z\"{u}rich, Z\"{u}rich, Switzerland}{daniel.rutschmann@inf.ethz.ch}{}{}

\authorrunning{S. Huang, C.-H. Liu, D. Rutschmann} 

\Copyright{Shengyu Huang, Chih-Hung Liu,Daniel Rutschman} 



\begin{document}

\maketitle

\begin{abstract}
Given $n$ elements, an integer $k$ and a parameter $\varepsilon$, 
we study to select an element with rank in $(k-n\varepsilon,k+n\varepsilon]$ using \emph{unreliable} comparisons where the outcome of each comparison is incorrect independently with a constant error probability, and multiple comparisons between the same pair of elements are independent. 
In this fault model, the fundamental problems of finding the minimum, selecting the $k$-th smallest element and sorting have been shown to require $\Theta\big(n \log \frac{1}{Q}\big)$, $\Theta\big(n\log \frac{\min\{k,n-k\}}{Q}\big)$ and $\Theta\big(n\log \frac{n}{Q}\big)$ comparisons, respectively, to achieve success probability $1-Q$~\cite{FeigeRPU94}.
Although finding the minimum and selecting the $k$-th smallest element have different complexities, to attain the high probability guarantee ($Q=\frac{1}{n}$), both of them require $\Theta(n\log n)$ comparisons.
Recently, Leucci and Liu~\cite{LeucciL20} proved that the approximate minimum selection problem ($k=0$) requires expected $\Theta(\varepsilon^{-1}\log \frac{1}{Q})$ comparisons. 
Therefore, it is interesting to study if there exists a clear distinction between the two problems in the approximation scenario.

 We develop a randomized algorithm that performs \emph{expected} $O(\frac{k}{n}\varepsilon^{-2} \log \frac{1}{Q})$ comparisons to achieve success probability at least $1-Q$.
We also prove that any randomized algorithm with success probability at least $1-Q$ performs \emph{expected} $\Omega(\frac{k}{n}\varepsilon^{-2}\log \frac{1}{Q})$ comparisons.
Our results indicate a clear distinction between approximating the minimum and approximating the $k$-th smallest element, which holds even for the high probability guarantee, e.g., if $k=\frac{n}{2}$ and $Q=\frac{1}{n}$, $\Theta(\varepsilon^{-1}\log n)$ versus $\Theta(\varepsilon^{-2}\log n)$.
Moreover, if $\varepsilon=n^{-\alpha}$ for $\alpha \in (0,\frac{1}{2})$,
the asymptotic difference is almost quadratic, i.e., $\tilde{\Theta}(n^{\alpha})$ versus $\tilde{\Theta}(n^{2\alpha})$.
As a by-product, we give an algorithm using deterministic $O\big(\frac{k}{n}\varepsilon^{-2}\log \frac{1}{Q}+(\log\frac{1}{Q})(\log \log\frac{1}{Q})^2\big)$ comparisons, which is optimal as long as $\frac{k}{n}\varepsilon^{-2}=\Omega\big((\log \log\frac{1}{Q})^2\big)$.
\end{abstract}

\newpage

\section{Introduction}\label{sec-ind}

We study a generalization of the fundamental problem of selecting \emph{the $k$-th smallest elements} in terms of \emph{approximation} and \emph{fault tolerance}.
Given a set $S$ of $n$ elements, an integer $k$ and a parameter $\varepsilon$, the \emph{fault-tolerant $\varepsilon$-approximate $k$-selection} problem, $\FT(k,\varepsilon)$ for short, is to return an element with rank in $(k-n\varepsilon,k+n\varepsilon]$ only using \emph{unreliable} comparisons whose outcome can be \emph{incorrect}.
Due to these comparison faults, it is impossible to guarantee a correct solution, so the number of comparisons performed by an algorithm should depend on the \emph{failure probability} $Q$ of the algorithm where $Q<\frac{1}{2}$.
Without loss of generality, we assume that $n$ is even and $k\leq \frac{n}{2}$; if $k>\frac{n}{2}$, the problem becomes to approximate the $(n-k)$-th largest element, which is symmetric.
The elements with rank in $(0, k-n\varepsilon]$, $(k-n\varepsilon,k+n\varepsilon]$ and $(k+n\varepsilon, n]$ of $S$ are called \emph{small}, \emph{relevant} and \emph{large}, respectively.

We consider \emph{independent random comparison faults}:
There is a strict ordering relation among $S$, but algorithms can only gather information via unreliable comparisons between two elements.
The outcome of each comparison is wrong with a known constant probability $p<\frac{1}{2}$.
When comparing the same pair of elements multiple times,
each outcome is \emph{independent} of the previous outcomes; 
comparisons involving different pairs of elements are also independent.

The above fault model has been widely studied for various fundamental problems such as finding the minimum, selecting the $k$-th smallest element and sorting a sequence~\cite{FeigeRPU94,Pelc89,Pelc02}. 
Feige et al~\cite{FeigeRPU94} proved that to achieve success probability $1-Q$, the aforementioned three problems require $\Theta\big(n \log \frac{1}{Q}\big)$, $\Theta\big(n\log \frac{\min\{k,n-k\}}{Q}\big)$ and $\Theta\big(n\log \frac{n}{Q}\big)$ comparisons, respectively, both in expectation and in the worst case.  
In the sequel, their selection algorithm is denoted by $\Sel(k,Q)$, and its performance is summarized as follows.

\begin{theorem}[\cite{FeigeRPU94}]\label{thm-feige-selection}
	$\Sel(k,Q)$ performs $O\big(n\log \frac{\min\{k,n-k\}}{Q}\big)$ comparisons to select the $k$-th smallest element among $n$ elements with probability at least $1-Q$. 
\end{theorem}

Due to the increasing complexity of modern computing, 
error detection and correction require enormous computing resources.  
Emerging technologies enable the tolerance of computation errors for saving computing resources~\cite{PalemL13,HanO13,ChoLM12,KirschP12,SloanSK12}.
Meanwhile, many practical applications do not require an optimal answer but good enough ones.
Therefore, fault-tolerant approximation algorithms are well-motivated.

An intuitive approach to the $\FT(k,\varepsilon)$ problem is first to pick $m=\Theta(\frac{k}{n}\varepsilon^{-2}\log\frac{1}{Q})$ elements randomly so that the underlying $\lceil k\cdot \frac{m}{n} \rceil$-th smallest element is \emph{relevant} with probability at least $1-\frac{Q}{2}$, and then to apply $\Sel(\lceil k\cdot \frac{m}{n} \rceil,\frac{Q}{2})$ on the $m$ elements.
By Theorem~\ref{thm-feige-selection}, this approach requires $\Theta\big(\frac{k}{n}\varepsilon^{-2}(\log^2\frac{1}{Q}+(\log\frac{1}{Q})(\log\frac{k}{n}\varepsilon^{-2}))\big)$ comparisons.
Recently, Leucci and Liu~\cite{LeucciL20} studied the approximate minimum selection problem, which asks for one element with rank in $(0, n\varepsilon]$ and thus is equivalent to $\FT(0,\varepsilon)$.
 They developed an algorithm using \emph{expected} $O(\varepsilon^{-1}\log\frac{1}{Q})$ comparisons and also proved a matching lower bound.

It is of great interest to study if the $\FT(k,\varepsilon)$ problem can be solved
with probability $1-Q$ using $O(\frac{k}{n}\varepsilon^{-2}\log \frac{1}{Q})$ comparisons. 
Moreover, although finding the minimum and finding the $k$-th smallest element require different numbers of comparisons, i.e., $\Theta(n\log\frac{1}{Q})$ versus $\Theta(n\log \frac{\min\{k,n-k\}}{Q})$, to attain the so-called high probability guarantee, i.e., $Q=\frac{1}{n}$, both problems require $\Theta(n\log n)$ comparisons. 
Thus, it is also desirable to investigate if there is a stronger distinction between these two problems in the approximation scenario.

\begin{remark}\label{rm-difficulty}
	Similar to many randomized algorithms, a bound with a $\log^2\frac{1}{Q}+(\log\frac{1}{Q})\cdot(\log \frac{k}{n}\varepsilon^{-2})$ term can be easily attained as in the above intuitive approach, but improving such a term to exactly $\log \frac{1}{Q}$ would be nontrivial.
	For example, Section~\ref{sec-median} will discuss how a variant of Quickselect fails to attain the $\log\frac{1}{Q}$ bound.
\end{remark}

\subsection{Our Contributions}\label{sub-contribution}

We develop a randomized algorithm that performs \emph{expected} $O(\frac{k}{n}\varepsilon^{-2}\log \frac{1}{Q})$ comparisons to solve the $\FT(k,\varepsilon)$ problem with probability at least $1-Q$.
We also prove that any algorithm with success probability $1-Q$ requires \emph{expected} $\Omega(\frac{k}{n}\varepsilon^{-2}\log \frac{1}{Q})$ comparisons, implying the optimality of our algorithm.
As a by-product, we give a randomized algorithm using deterministic $O\big(\frac{k}{n}\varepsilon^{-2}\log \frac{1}{Q}+(\log\frac{1}{Q})(\log\log\frac{1}{Q})^2\big)$ comparisons, which is optimal as long as $\frac{k}{n}\varepsilon^{-2}=\Omega\big((\log \log\frac{1}{Q})^2\big)$.

Our results indicate that there is a distinction between the approximate \emph{minimum} selection problem and the general approximate \emph{$k$-th element} selection problem in terms of the expected number of comparisons, i.e., $\Theta(\varepsilon^{-1}\log\frac{1}{Q})$~\cite{LeucciLM19} versus $\Theta(\frac{k}{n}\varepsilon^{-2}\log \frac{1}{Q})$.
This distinction even holds for the high probability guarantee ($Q=\frac{1}{n}$) in contradiction to the fact that the two problems have the same complexity $\Theta(n\log n)$ in the \emph{exact} selection~\cite{FeigeRPU94}.
For example,  if $k=\frac{n}{2}$ and $Q=\frac{1}{n}$, the two approximate selection problems require expected $\Theta(\varepsilon^{-1}\log n)$ and  
$\Theta(\varepsilon^{-2}\log n)$ comparisons, respectively.
Moreover, if $\varepsilon=n^{-\alpha}$ for a constant $\alpha \in (0,\frac{1}{2})$,
the asymptotic difference is almost quadratic, i.e., $\tilde{\Theta}(n^{\alpha})$ versus $\tilde{\Theta}(n^{2\alpha})$.

\begin{remark}\label{rm-ignorance}
The $\frac{k}{n}\varepsilon^{-2}$ term in those complexities is actually $\max\{\varepsilon^{-1}, \frac{k}{n}\varepsilon^{-2}\}$. If $k\leq n\epsilon$, by which $\varepsilon^{-1}\geq \frac{k}{n}\varepsilon^{-2}$, a correct answer to $\FT(0,\varepsilon)$ is also correct to $\FT(k,\varepsilon)$, indicating that this case is essentially the approximate minimum selection and can be solved optimally by Leucci and Liu's algorithms~\cite{LeucciL20}.
Therefore, to simplifying the description, we assume that $k>n\epsilon$ throughout the paper if no further specification.
\end{remark}

As noted in Remark~\ref{rm-difficulty}, our technical advance is to improve the $\log^2\frac{1}{Q}+(\log\frac{1}{Q})(\log \frac{k}{n}\varepsilon^{-2})$ term to $\log\frac{1}{Q}$.
To some extent, compared with Leucci and Liu's algorithms, our algorithms cover the entire range of $k$ instead of the case when $k$ is trivially small.
In addition, our algorithm owns an elegant feature that it only exploits simple sampling techniques, e.g., selecting the median of three samples and selecting the minimum of two samples.

The top-level of our algorithm, inspired by Leucci and Liu~\cite{LeucciL20}, reduces the $\FT(k,\varepsilon)$ problem on $n$ elements
to the $\FT(\frac{m}{2}, \frac{3}{8} )$ problem on $m=\Theta(\log\frac{1}{Q})$ elements.
More precisely, 
if a relevant element can be selected with probability  $\frac{8}{9}$,
we can generate a sequence of $\Theta(\log\frac{1}{Q})$ elements in which $\frac{3}{4}$ of elements around the middle, with probability $1-\frac{Q}{2}$, are all relevant.
For such a ``dense'' sequence, we design a delicate trial-and-error method to select a relevant element with probability $1-\frac{Q}{2}$ using expected $\Theta(\log\frac{1}{Q})$ comparisons.

The main challenge is to obtain a relevant element with probability $\frac{8}{9}$ using only  $O(\frac{k}{n}\varepsilon^{-2})$ comparisons.
For the approximate minimum ($k=0$), Leucci and Liu~\cite{LeucciL20} applied $\Sel(1,\frac{1}{10})$ on $\Theta(\varepsilon^{-1})$ randomly picked elements and attained $O(\varepsilon^{-1})$ comparisons.
However, for general $k$, this method requires $\Theta(\frac{k}{n}\varepsilon^{-2} \log \frac{k}{n}\varepsilon^{-2} )$ comparisons with an extra logarithmic factor.

We first work on a special case that $k=\frac{n}{2}$, i.e., the approximate median selection. 
Based on the symmetry property of the median, we observe that the median of three randomly picked elements is more likely to be relevant than a randomly picked element.
We exploit this observation to iteratively increase the ratio of relevant elements while keeping the underlying median being relevant.
Once the ratio becomes a constant fraction, we will apply a straightforward method.

For general $k$, we design a ``purifying'' process that iteratively increases the ratio of relevant elements while keeping elements around a ``controlled'' position being relevant.
Despite no symmetry property, we still observe that under certain conditions, the minimum of two randomly picked elements is more likely to be relevant than a randomly picked one.
Then, we derive feasible parameters to control the relative position of $k$, i.e., the middle of the remaining relevant elements, during the purifying process.
Once the relative position becomes a constant fraction of the remaining elements, we add dummy smallest elements and apply our approximate median selection.

For some range of $(k,\varepsilon)$, our bounds are not tight.
If $\frac{k}{n}\varepsilon^{-2}\log \frac{1}{Q}=\Omega(n)$, the lower bound is $\Omega(\max\{n, \varepsilon^{-1}\log\frac{(k+n\varepsilon)/(2n\varepsilon)}{Q} \})$ (Theorem~\ref{thm-lower-another}).
For this range, a trivial upper bound of $O(n\log \frac{k}{Q})$ follows from Theorem~\ref{thm-feige-selection},
indicating a gap between $\Omega(\max\{n, \varepsilon^{-1}\log(\frac{k+n\varepsilon}{2n\varepsilon}\cdot \frac{1}{Q}) \})$ and $O(n\log \frac{k}{Q})$ for some range of $(k,\varepsilon)$.

The rest of the paper is organized as follows. 
Section~\ref{sub-literature} gives a brief literature review.
Section~\ref{sec-pre} provides a few preliminary remarks.
Section~\ref{sec-reduction} presents the top-level algorithm.
Section~\ref{sec-median} and Section~\ref{sec-selection} describe sub-algorithms to approximate the median and the $k$-th element with constant probability, respectively.
Section~\ref{sec-lower} sketches the lower bound analysis. Interested readers are referred to the appendix for detailed technical proofs.

\subsection{Brief Literature}\label{sub-literature}

Dating back to the 1987, 
Ravikumar et al.~\cite{RavikumarGL87} already studied a variant of the problem of finding the \emph{exact} minimum using unreliable comparisons when at most $f$ comparisons are allowed.
They proved that $\Theta(f n)$ comparisons are necessary in the worst case. 
Later, Aigner~\cite{Aigner97} considered a \emph{prefix-bounded} error model: for a fraction parameter  $\gamma < \frac{1}{2}$, at most an $\gamma$-fraction of the past comparisons failed at any point during the execution of an algorithm. 
He proved that $\Theta(\frac{1}{1-p})^n$ comparisons is necessary to find the minimum in the worst case. 
Furthermore, he proved that if $p>\frac{1}{n-1}$, no algorithm can succeed with certainty~\cite{Aigner97}.

When errors occur independently,
as already discussed,
Feige et al.~\cite{FeigeRPU94} showed that the required number of comparisons for selecting the exact $k$-th smallest element with probability at least $1-Q$ is $\Theta(n\log \frac{\max\{k,n-k\}}{Q})$. 
Recently, Braverman et al.~\cite{BravermanMW16} investigated the \emph{round complexity} and the number of comparisons required by partition and selection algorithms. 
They proved that for any constant error probability,
$\Theta(n \log n)$ comparisons are necessary for any algorithm that selects the minimum with high probability.
Also, Chen et al. \cite{ChenGMS17} studied the problem of computing the smallest $k$ elements using $r$ given independent noisy comparisons between each pair of elements. 
In a very general error model called \emph{strong stochastic model},  they gave a linear-time algorithm with competitive ratio of $\tilde{O}(\sqrt{n})$, and also proved that this competitive ratio is tight.

The related problem of sorting with faults has also received considerable attention.
When there are at most $f$ comparison faults,
$\Theta(n \log n + fn)$ comparisons are necessary and sufficient to correctly sort $n$ elements~\cite{LakshmananRG91, Long92, Bagchi92}. 
For the prefix-bounded model, 
although Aigner's result on the minimum selection \cite{Aigner97} implies that $(\frac{1}{1-p})^{O(n \log n)}$ are sufficient to sort $n$ elements, 
Borgstrom and Kosaraju~\cite{BorgstromK93} showed that
checking whether the input elements are sorted already requires $\Omega\big( ( \frac{1}{1-p} )^n \big)$ comparisons.
When comparison faults are permanent, or equivalently, when a pair of elements can only be compared once,
the underlying sorting problem has also been extensively studied especially because it can be connected to both the \emph{minimum feedback arc set}  problem and the \emph{rank aggregation} problem \cite{MakarychevMV13, Kenyon-MathieuS07, BravermanMW16, BravermanM08, KleinPSW22,LeightonM99, GeissmannMW15, GeissmannLLP17, GeissmannLLP20, GeissmannLLP19}.  
There are also sorting algorithms for memory faults~\cite{FinocchiGI09,LeucciLM19}.

For more knowledge about fault-tolerant search algorithms, we refer the interested readers to a survey by Pelc~\cite{Pelc02} and a monograph by Cicalese~\cite{Cicalese13}.

\section{Preliminary}\label{sec-pre}

As explained in remark~\ref{rm-ignorance}, we assume that $k>n\epsilon$ throughout the paper if no further specification.
For ease of exposition, we use $\bm{\beta}$ to denote $\bm{\frac{k}{n}}$ in some analyses 
and sometimes abuse the name $x$ of an element to denote its rank, e.g., we might write ``$x\in [l, r]$'' to denote that the rank of $x$ lies in the range $[l, r]$. 
Comparing two elements, $x$ and $y$, yields an outcome of either $x<y$ or $y>x$.
A typical subroutine in our algorithms is to draw elements using sampling with replacement,
so multiple copies of an element may appear in a set.
When two copies of the same
element are compared, the tie is broken using any arbitrary (but consistent) ordering among the copies.

In our fault model, there is a standard strategy called \emph{majority vote} for reducing the ``error probability'' of comparing two elements.
We state this strategy as follows.

\newcommand{\lemmajorityvote}{\emph{(Majority Vote)}
For any error probability $p \in [0, \frac{1}{2} )$, there exists a postive integer $\MV_p$ such that a strategy that compares two elements $2 \MV_p \cdot t+1$ times and returns the majority result succeeds with probability at least $1- e^{-t}$, where $c_p = \lceil \frac{4(1-p)}{(1-2p)^2} \rceil$.
The exact failure probability of this strategy is 
\begin{equation*}\sum_{i=0}^{\MV_p\cdot t} {2\MV_p\cdot t+1\choose i}(1-p)^i p^{2\MV_p\cdot t+1-i}.\end{equation*}}

\begin{lemma}\label{lem-majority-vote}
\lemmajorityvote
\end{lemma}

\section{Top Level of Algorithm}\label{sec-reduction}

The high-level idea is to reduce solving $\FT(k,\varepsilon)$ on $n$ elements with probability at least $1-Q$ to solving $\FT(\frac{m}{2},\frac{3}{8})$ on $m=\Theta(\log\frac{1}{Q})$ elements with probability at least $1-\frac{Q}{2}$.
Specifically, if a \emph{relevant} element can be selected with probability at least $\frac{8}{9}$,
then $m$ selected elements, for some $m=\Theta(\log \frac{1}{Q})$, contain at least $\frac{7}{8}m$ relevant elements with probability at least $1-\frac{Q}{2}$; see Lemma~\ref{lem-log-Q-sample} in Appendix~\ref{ap-reduction}. 
In this situation, at least $2\cdot(\frac{7}{8}-\frac{1}{2})\cdot m = 2\cdot \frac{3}{8} m$ elements around the median, i.e., the range $(\frac{1}{8}m,\frac{7}{8}m]$, are relevant.
Therefore, solving the $\FT(\frac{m}{2},\frac{3}{8})$ problem on these $m$ elements with probability at least $1-\frac{Q}{2}$ yields a relevant element with probability at least $1-2\cdot \frac{Q}{2}=1-Q$.

Section~\ref{sec-selection} will present an approach that uses $O(\frac{k}{n}\varepsilon^{-2})$ comparisons to select a relevant element with probability at least $\frac{8}{9}$,
by which the above reduction takes $O(\frac{k}{n}\varepsilon^{-2}\log\frac{1}{Q})$ comparisons. 
In the remaining of this section, we will explain how to solve $\FT(\frac{m}{2},\frac{3}{8})$ with probability $1-\frac{Q}{2}$ efficiently both in expectation and in determination cases.

We first design a simple trial-and-error method that uses \emph{expected} $O(\log \frac{1}{Q})$ comparisons to select an element from $(\frac{1}{8}m,\frac{7}{8}m]$ with probability at least $1-\frac{Q}{2}$:
\begin{quote} 
	Repeatedly pick an element randomly and verify if its rank lies in $(\frac{1}{8}m,\frac{7}{8}m]$ until one element passes the verification.
\end{quote}
Since $(\frac{1}{8}m,\frac{7}{8}m]$ contains $\frac{3}{4}m$ elements, 
the expected number of repetitions before encountering a correct element is only $O(1)$. 
Therefore, the key is to implement the verification step such that the method returns a correct element with probability at least $1-\frac{Q}{2}$ and the expected number of comparisons is $O(\log \frac{1}{Q})$.

We implement the \emph{verification step} for an element $x$ based on a simple experiment that randomly picks three other elements, and checks if $x$ is neither the smallest nor the largest among the four elements.
The probability that the if-condition holds is $1-(\frac{r_x}{m})^3-(1-\frac{r_x}{m})^3$ where $r_x$ is the rank of $x$ among the $m$ elements.
Also, the check can be conducted with success probability at least $\frac{17}{18}$ using $O(1)$ comparisons (by plugging in $n=4$, $Q=\frac{1}{36}$ into Theorem~\ref{thm-feige-selection} twice with $k=1$ and $k=4$.)
Therefore, if $x\in (\frac{2}{8}m,\frac{6}{8}m]$, the experiment succeeds with probability \emph{at least} $\frac{9}{16}\cdot\frac{17}{18}=\frac{17}{32}$,
while if $x\in [1,\frac{1}{8}m]$ or $x\in (\frac{7}{8}m, m]$,
the experiment succeeds with probability \emph{at most} $\frac{21}{64}+\frac{1}{18}=\frac{221}{576}\leq \frac{15}{32}$.

In the above derivation, we ignore two ranges $(\frac{1}{8}m,\frac{2}{8}m]$ and $(\frac{6}{8}m,\frac{7}{8}m]$ since returning an element in these two ranges is safe and the considered range $(\frac{2}{8}m,\frac{6}{8}m]$ contains enough elements.
Based on the above calculated probabilities, we can conceptually treat the above simple experiment as an unreliable comparison with error probability $\frac{15}{32}$.
By Lemma~\ref{lem-majority-vote}, if the verification step conducts this simple experiment $2\cdot\MV_{15/32}\ln \frac{2}{Q}+1$ times and takes the majority result, its success probability is at least $1-\frac{Q}{2}$,

Now, we are ready to analyze the expected number of comparisons and the success probability of our trial-and-error method.
First, a single round returns an element in $(\frac{2}{8}m,\frac{6}{8}m]$ with probability \emph{at least} $\frac{1}{2}\cdot (1-\frac{Q}{2})\geq \frac{1}{4}$, and thus the probability to conduct the $i$-th round is at most $(\frac{3}{4})^{i-1}$.
Therefore, the expected number of comparisons is at most $\sum_{i\geq 1}(\frac{3}{4})^{i-1}\cdot (2\cdot\MV_{15/32}\ln \frac{2}{Q}+1)=O(\log \frac{1}{Q})$. 
Moreover, 
a single round returns an element in $[1,\frac{1}{8}m]$ or $(\frac{7}{8}m, m]$ with probability \emph{at most} $\frac{1}{4}\cdot \frac{Q}{2}=\frac{Q}{8}$, 
so the failure probability is at most $\sum_{i\geq 1}(\frac{3}{4})^{i-1}\cdot \frac{Q}{8}= \frac{Q}{2}$, concluding the following theorem:

\newcommand{\thmexpected}{It takes \textbf{expected} $O(\frac{k}{n}\varepsilon^{-2}\log\frac{1}{Q})$ comparisons to solve the \mbox{$\FT(k,\varepsilon)$} problem with probability at least $1-Q$.}

\begin{theorem}\label{thm-expected}
	\thmexpected
\end{theorem}

Finally, to derive a deterministic bound, we note that the simple experiment in the
verification step may be viewed as a \emph{biased} coin toss. 
From this viewpoint, we are able to turn the $\FT(\frac{m}{2},\frac{3}{8})$ problem into finding a coin
with bias bigger than $\frac{15}{32}$, given that at least half of the coins have bias at least $\frac{17}{32}$. 
Grossman and Moshkovitz \cite{GrossmanM20} provided an algorithm that solves the new problem with probability $1-\frac{Q}{2}$ using $O(\log \frac{1}{Q} \cdot (\log \log \frac{1}{Q})^2)$ coin tosses,
leading to the following theorem.

\newcommand{\thmdeterminsticfaster}{It takes $O(\frac{k}{n}\varepsilon^{-2}\log\frac{1}{Q}+\log\frac{1}{Q} (\log \log \frac{1}{Q})^2)$ comparisons to solve the $\FT(k,\varepsilon)$ problem with probability at least $1-Q$.}

\begin{theorem}\label{thm-determinstic-faster}
	\thmdeterminsticfaster
\end{theorem}

\section{Approximate Median Selection}\label{sec-median}

We attempt to select an element in $(\frac{n}{2}-n\varepsilon,\frac{n}{2}+n\varepsilon]$, i.e., $\bm{k=\frac{n}{2}}$, with probability at least $1-\frac{1}{18}$ using only $O(\varepsilon^{-2})$ comparisons. This algorithm will then be applied in Section~\ref{sec-selection} as a subroutine.
A straightforward method, denoted by $\SM(\varepsilon)$, picks $m=\Theta(\varepsilon^{-2})$ elements randomly to make their median \emph{relevant} with probability at least $1-\frac{1}{72}$ and applies the $\Sel(\frac{m}{2},\frac{1}{72})$ algorithm (Theorem~\ref{thm-feige-selection}), resulting in a failure probability of at most $\frac{1}{36}$.
However, the $\Sel(\frac{m}{2},\frac{1}{72})$ algorithm takes $O(m\log\frac{m}{1/72})=O(\varepsilon^{-2}\log\varepsilon^{-1})$ comparisons with an \emph{extra} logarithmic factor.
To achieve $O(\varepsilon^{-2})$ comparisons,
we will ``purify'' the input elements in a way that the ratio of relevant elements is increasing while the underlying median is still relevant.
Once the ratio of relevant elements becomes a constant fraction, i.e., from $2\varepsilon$ to $O(1)$, we can afford to apply the $\SM$ algorithm.
We assume that $\varepsilon<\mb$ since if $\varepsilon\geq \mb$, the $\SM(\varepsilon)$ algorithm takes only $O(\varepsilon^{-2}\log\varepsilon^{-1})=O(1)$ comparisons.

A major difficulty to overcome in the purifying process is the following: 
if we consider three elements that are each relevant with probability
$\rho$, then, even in the absence of comparison faults,
their median is relevant with probability at most $\frac{3}{2} \rho + O(\rho^2)$,
which is a lot less than $3 \rho$.
Thus, one risks running out of elements long before the ratio of relevant elements becomes a constant.
This issue remains if we replace three by a larger constant,
and it applies to any algorithm that works in a non-constant number of phases,
including algorithms that more closely resemble Quickselect. Those algorithms
would need to start with $\Omega(\varepsilon^{-(2 + \delta)})$ elements
for some $\delta > 0$ and hence cannot achieve the $O(\varepsilon^{-2})$ bound.

To settle the above issue, we maintain a multiset of elements and re-sample from this multiset at every phase. 
Our re-sampling method allows us to decrease the number of elements by less than a factor of $\frac{3}{2}$, so we can avoid running out of elements.

The algorithm is sketched as follows:
\begin{enumerate}
	\item For $1\leq i\leq L$, generate a multiset
	 $M_i$ of $n_i$ elements by repeatedly picking three elements from $M_{i-1}$ \emph{randomly} and selecting the median of the three using a symmetric median selection algorithm (Lemma~\ref{lem-symmetry-median} below).
	\item Apply the $\SM(\bm{\varepsilon_L})$ algorithm on $M_L$.
\end{enumerate}
Initially, $M_0=S$, $n_0=n$, $\varepsilon_0=\varepsilon$. $M_i$ is called \emph{\textbf{good}} if all elements in the range $(\frac{n_i}{2}-n_i\varepsilon_i,\frac{n_i}{2}+n_i\varepsilon_i]$ are \emph{relevant}.
Moreover, $n_i$ is decreasing with $i$ while $\varepsilon_i$ is increasing with $i$, and $L=\min\{i\mid \varepsilon_{i}\geq \mb\}$, i.e., the minimum of number of rounds such that at least $2\cdot\mb=\frac{1}{3}$ of the elements around the middle is relevant.
The rest of this section illustrates the idea behind this process and implements these parameters $n_i$ and $\varepsilon_i$.

\newcommand{\lemsymmetrymedian}{For \textbf{three} elements, consider the following median selection algorithm:
	\begin{enumerate}
		\item For each pair of elements, apply the majority vote strategy with $2c_p\cdot 4+1$ comparisons (Lemma~\ref{lem-majority-vote}), and assign a point to the element that attains the majority result.
		\item Return the element with exactly one point. If all three elements get exactly one point, return one of them uniformly at random.
	\end{enumerate}
	The above algorithm returns the median with probability at least $1-\frac{1}{13}$, and returns the minimum and the maximum with the same probability, i.e., at most $\frac{1}{26}$.}
\begin{lemma}\label{lem-symmetry-median}
	\lemsymmetrymedian
\end{lemma}

\newcommand{\lemselectedmedianthree}{If $M_{i-1}$ is good, then each element in $M_i$ is small (resp. large) with probability at most $\frac{1}{2}-\frac{4}{3}\varepsilon_{i-1}$.}

The purifying process is inspired by a simple observation:
a randomly picked element is relevant with probability $2\varepsilon$, while the \emph{median} of three randomly picked elements is relevant with probability much greater than $2\varepsilon$.
Let $E_S$ denote the event that the median of  three randomly picked elements is small. Then,
\[\Pr[E_S]=3\left(\frac{1}{2}-\varepsilon \right)^2\left(\frac{1}{2}+\varepsilon \right)+\left(\frac{1}{2}-\varepsilon\right)^3=\frac{1}{2}-\frac{3}{2}\varepsilon+2\varepsilon^3.\]
If $\varepsilon<\mb$, then $\Pr[E_S]\leq \frac{1}{2}-\frac{3}{2}\varepsilon+2(\mb)^2\varepsilon=1-\frac{13}{9}\varepsilon$.
By Lemma~\ref{lem-symmetry-median},
the median selection returns the median with probability at least $1-\frac{1}{13}$, and returns the minimum (resp. the maximum) with probability at most $\frac{1}{26}$. 
A simple calculation, together with the above arguments, gives the following lemma:
\begin{lemma}\label{lem-selected-median-three}
	\lemselectedmedianthree
\end{lemma}
By Lemma~\ref{lem-selected-median-three}, it is feasible to set $\varepsilon_i=(\frac{5}{4})^{i}\cdot \varepsilon$, i.e., growing slightly slower than $\frac{4}{3}$.

The size $n_i$ is set as $\lceil 2000 \cdot i \cdot (\frac{4}{5})^{\bm{2i}} \cdot \varepsilon^{-2}\rceil$ to limit the number of comparisons and the failure probability. 
First, $n_i$ is linear in $\varepsilon^{-2}$ since the minimum number of elements to be looked at is $\Omega(\varepsilon^{-2})$ (Section~\ref{sec-lower}).
Second, to bound the total number of comparisons, $n_i$ should shrink exponentially with $i$.
Third, to bound the failure probability of the algorithm, the failure probability of the $i$-th round should also shrink exponentially with $i$. 
From the above three aspects, since the Chernoff bound (Lemma~\ref{lem-chernoff} in Appendix~\ref{ap-pre}) will be applied for the probabilitic analysis, $n_i$ should be linear in $i$, and the shrink factor of $n_i$ should be at least $(\frac{4}{5})^2$ to cancel out the square of the growth factor $\frac{5}{4}$ of $\varepsilon_i$.

\newcommand{\lemgoodmedian}{For $1\leq i\leq L$
	\[\Pr[M_i \mbox{ is NOT good }\mid M_{i-1} \mbox{ is good }]\leq 2\cdot e^{-5i}.\]}

Because the $\SM(\varepsilon_L)$ algorithm fails with probability at most $\frac{1}{36}$, it is sufficient to prove that $\Pr[M_L \mbox{ is good }]\geq 1-\frac{1}{36}$.
Let $E_i$ denote the event that $M_i$ is good.
By definition, $\Pr[E_0]=1$. With the Chernoff bound, we can prove the following lemma:
\begin{lemma}\label{lem-good-median}
	\lemgoodmedian 
\end{lemma}
By Lemma~\ref{lem-good-median}, we can lower bound $\Pr[E_L]$ as 
\[\Pr[E_L]=1-\Pr[\bigcup_{i=1}^L\overline{E_i}\mid E_{i-1}]\geq 1-\sum_{i=1}^L 2\cdot e^{-5i}\geq 1-4\cdot e^{-5}\geq 1- \frac{1}{36}.\]
By Lemma~\ref{lem-symmetry-median}, each median selection takes $O(1)$ comparisons, so the purifying process takes $O(\sum_{i=1}^L n_i)=O\big(\varepsilon^{-2}\sum_{i=1}^L i\cdot(\frac{4}{5})^{2i}\big)=O(\varepsilon^{-2})$ comparisons.
Since $\varepsilon_L\geq \mb$, the $\SM(\varepsilon_L)$ algorithm takes $O(1)$ comparisons, concluding the following theorem:

\newcommand{\thmmedianconstant}{It takes $O(\varepsilon^{-2})$ comparisons to select an element in $(\frac{n}{2}-n\varepsilon,\frac{n}{2}+n\varepsilon]$ with probability at least $1-\frac{1}{18}$.}

\begin{theorem}\label{thm-median-constant}
	\thmmedianconstant
\end{theorem}

\section{Approximate $k$-th Element Selection}\label{sec-selection}

We attempt to select an element in $(k-n\varepsilon, k+n\varepsilon]$ with probability at least $1-\frac{1}{9}$ using only $O(\frac{k}{n}\varepsilon^{-2})$ comparisons. 
Recall that $k> n\varepsilon$ as assumed in Remark~\ref{rm-ignorance}.
 If $n\varepsilon<k\leq 2n\varepsilon$, we halve the value of $\varepsilon$ so that $k > 2n\varepsilon$, which does not increase the asymptotic complexity.
Therefore, we can safely assume $k>2n\varepsilon$ afterwards. 
In this scenario, the straightforward approach mentioned in Section~\ref{sec-ind} requires $O\big(\frac{k}{n}\varepsilon^{-2}\log (k\varepsilon^{-1})\big)$ comparisons with an extra $\log (k\varepsilon^{-1})$ factor.
Another approach is to add $n-2k$ dummy smallest elements (so that the relevant elements lie in the middle) and to apply the algorithm in Section~\ref{sec-median} with $\frac{\varepsilon}{2}$,
leading to $O(\varepsilon^{-2})$ comparisons. 
As a result, both approaches are more expensive than $O(\frac{k}{n}\varepsilon^{-2})$.

At a high level,
our breakthrough is an iterative ``purifying'' process
that increases both the ratio of relevant elements and the relative position of $k$, i.e., the middle position of relevant elements, while ``controlling'' the relative position.
Once the relative position becomes a constant fraction of the remaining elements, e.g., $\frac{1}{8}$, we add dummy smallest elements and apply the approximate median selection algorithm in Section~\ref{sec-median}. 
As the ratio of relevant elements increases at the same time,
the resulting number of comparisons will be $O(\frac{k}{n}\varepsilon^{-2})$ instead of $O(\varepsilon^{-2})$.

The algorithm is sketched as follows:
\begin{enumerate}
	\item For $1\leq i\leq L$, generate a set $S_i$ of $n_i$ elements by repeatedly picking two elements from $S_{i-1}$ \emph{randomly} and selecting the minimum of the two using $6\MV_p+1$ comparisons (Lemma~\ref{lem-majority-vote}).
	\item Add $n_L-2k_l+2\varepsilon_L$ dummy smallest elements to $M_L$ and apply the approximate median selection algorithm in Section~\ref{sec-median} on $M_L$ with respect to $\varepsilon_L$. 
\end{enumerate}
Initially, $S_0=S$, $n_0=n$, $k_0=k$, $\varepsilon_0=\varepsilon$. $S_i$ is called \emph{\textbf{good}} if all elements in the range $(k_i-n_i\varepsilon_i,k_i+n_i\varepsilon_i]$ are \emph{relevant}. 
For ease of exposition, let $\beta_i$ denote $\frac{k_i}{n_i}$.
Both $\beta_i$ and $\varepsilon_i$ increase with $i$ while $n_i$ decreases with $i$, and we set $L=\min\{i \mid  \beta_i \geq \frac{1}{8}\}$.
Recall that $\beta=\frac{k}{n}$. 
We assume that $\beta<\frac{1}{8}$; otherwise, we conduct the second step directly, i.e., $L=0$.

The purifying process is based on a simple observation that the minimum of two randomly picked element is \emph{small} with probability
\[\underbrace{\left(\beta-\varepsilon\right)^2}_{\mbox{two small}}+\underbrace{2\left(\beta-\varepsilon\right)\left(1-\left(\beta-\varepsilon\right)\right)}_{\mbox{one small \& one non-small}}=2\left(\beta-\varepsilon\right)-\left(\beta-\varepsilon\right)^2,\]
while a randomly picked element is small with probability merely $\beta-\varepsilon$. 
By a similar calculation, the minimum of two randomly picked elements is \emph{relevant} with $4\varepsilon-\beta\cdot 4\varepsilon$.
Since $k$ is exactly the number of small elements plus half the number of relevant elements, the above derivation suggests the following formulation of $\beta_i$:
\[\beta_i:=\underbrace{2\left(\beta_{i-1}-\varepsilon_{i-1}\right)-\left(\beta_{i-1}-\varepsilon_{i-1}\right)^2}_{\Pr[\mbox{ small }]}+\underbrace{\left(2\varepsilon_{i-1}-\beta_{i-1}\cdot 2\varepsilon\right)}_{\Pr[\mbox{ relevant }]\div 2}.\]

These derivations need to adapt to the failure probability $q$ of selecting the minimum using $6\MV_p+1$ comparisons.
By Lemma~\ref{lem-majority-vote}, $q\leq e^{-3}<\frac{1}{20}$ and $q=\sum_{i=1}^{3\MV_p}(1-p)^ip^{6\MV_p+1-i}$. 
Then, a selected element in the first round is \emph{relevant} with probability
\[\underbrace{4\varepsilon^2}_{
	\mbox{two relevant}}+q\cdot\underbrace{ 2\cdot \left(\beta-\varepsilon\right)2\varepsilon}_{\mbox{one small \& one relevant}}+(1-q)\cdot \underbrace{2\cdot \left(1-\left(\beta+\varepsilon\right)\right)2\varepsilon}_{\mbox{one large \& one relevant}},\]
which is equal to $4\varepsilon\cdot \big((1-q)-(1-2q)\cdot\beta\big)$.
Since $\beta<\frac{1}{8}$ and $q<\frac{1}{20}$, the above probability is larger than $\frac{67}{40}\cdot 2\varepsilon$. 
Therefore, it is feasible to set $\bm{\varepsilon_{i}=(\frac{3}{2})^i\cdot \varepsilon}$,  i.e., growing slower than $\frac{67}{20}$.

\newcommand{\lemselectioninvariant}{For $0\leq i\leq L$, \[\beta_i > 2\varepsilon_i \mbox{\quad and\quad} \beta_{i}\leq 2^i\cdot \beta. \mbox{\quad Thus,\quad} \frac{k_i}{n_i}\leq 2^i\cdot \frac{k}{n}\mbox{\quad for\quad} 0\leq i\leq L.\]}

To fit the formulation of $\beta_i$ to the above failure probability $q$, a similar calculation yields that 
each selected element in the first round is \emph{small} with probability
\[\left(\beta-\varepsilon\right)^2+(1-q)\cdot 2\left(\beta-\varepsilon\right)\left(1-\left(\beta-\varepsilon\right)\right).\]
Since the relative position is the number of small elements plus half the number of relevant elements, it is feasible to set the value of $\beta_i$ as follows (after arrangement):
\[\beta_i:= \left(2\beta_{i-1}-\beta_{i-1}^2-\varepsilon_{i-1}^2\right)-2q\left(\beta_{i-1}-\beta_{i-1}^2-\varepsilon_{i-1}^2\right).\]
Moreover, we can prove by induction important properties of $\beta_i$ as stated below:
\begin{lemma}\label{lem-selection-invariant}
	\lemselectioninvariant
\end{lemma}

The size $n_i$ of $S_i$ is set as $\lceil 960 \cdot i \cdot (\frac{8}{9})^i \cdot \frac{k}{n} \varepsilon^{-2}\rceil$ to control the number of comparisons and the failure probability.
Similar to Section~\ref{sec-median}, $n_i$ should shrink exponentially with $i$ and should also be linear in both $\frac{k}{n} \varepsilon^{-2}$ and $i$.
The major difference lies in that the existence of $k_i$ changes the shrink factor of $n_i$.
Since $\frac{k_i}{n_i}\leq 2^i\cdot \frac{k}{n}$ and $\varepsilon_i=(\frac{3}{2})^i\cdot \varepsilon$, the shrink factor of $n_i$ should be at least $\frac{8}{9}$.
This is based on the fact that $2^{-i}\cdot (\frac{3}{2})^{2i}\cdot (\frac{8}{9})^i=1$, which will be much clearer in the probability analysis.

To sum up, $\bm{n_i}=\lceil 960 \cdot i \cdot (\frac{8}{9})^i \cdot \frac{k}{n} \varepsilon^{-2}\rceil$, $\bm{\varepsilon_i}=(\frac{3}{2})^i\cdot \varepsilon$, $\bm{\beta_i}= (2\beta_{i-1}-\beta_{i-1}^2-\varepsilon_{i-1}^2)-2q\cdot(\beta_{i-1}-\beta_{i-1}^2-\varepsilon_{i-1}^2)$ with $\bm{q}=\sum_{i=1}^{3\MV_p}(1-p)^i p^{6\MV_p\cdot +1-i}$, and $\bm{L}=\min\{i \mid  \beta_i \geq \frac{1}{8}\}$.

\newcommand{\lemgoodselection}{For $1\leq i\leq L$
	\[\Pr[S_i \mbox{ is NOT good }\mid S_{i-1} \mbox{\; is good \;}]\leq 2\cdot e^{-4i}.\]}

To attain the success probability $1-\frac{1}{9}$, it is sufficient to prove that $\Pr[S_L\mbox{ is good}]\geq 1-\frac{1}{18}$ (Theorem~\ref{thm-median-constant}) since the approximate median selection in Section~\ref{sec-median} fails with probability at most $\frac{1}{18}$.
Let $E_i$ denote the event that $S_i$ is good.
By definition, $\Pr[E_0]=1$. Applying the Chernoff bound with the above parameters gives the following lemma:
\begin{lemma}\label{lem-good-selection}
	\lemgoodselection
\end{lemma}
By Lemma~\ref{lem-good-selection}, we can lower bound $\Pr[E_L]$ as 
\[\Pr[E_L]=1-\Pr[\bigcup_{i=1}^L\overline{E_i}\mid E_{i-1}]\geq 1-\sum_{i=1}^L 2\cdot e^{-4i}\geq 1-4\cdot e^{-4}\geq 1- \frac{1}{9}.\]

For the number of comparisons,
since each selection takes $6\MV_P+1=O(1)$ comparisons,
the purifying process takes $\sum_{i=1}^LO(n_i)=\frac{k}{n}\varepsilon^{-2}\cdot\sum_{i=1}^L O\big(i\cdot(\frac{8}{9})^i\big)=O(\frac{k}{n}\varepsilon^{-2})$ comparisons.
By Theorem~\ref{thm-median-constant}, the approximate median selection takes $O(\varepsilon_L^{-2})=O\big((\frac{2}{3})^{2L}\varepsilon^{-2}\big)=O(2^{-L}\cdot \varepsilon^{-2})$ comparisons.
Since $\frac{k_L}{n_L}\leq 2^L\cdot \frac{k}{n}$ (Lemma~\ref{lem-selection-invariant}) and $\frac{k_L}{n_L}>\frac{1}{8}$,
we have $2^{-L}=O(\frac{k}{n})$ and $O(2^{-L}\cdot \varepsilon^{-2})=O(\frac{k}{n}\varepsilon^{-2})$,
implying the following main theorem:

\newcommand{\thmselectionconstant}{It takes $O(\frac{k}{n}\varepsilon^{-2})$ comparisons to select an element in $(k-n\varepsilon,k+\varepsilon]$ with probability at least $1-\frac{1}{9}$.}

\begin{theorem}\label{thm-selection-constant}
	\thmselectionconstant 
\end{theorem}


\section{Lower Bound}\label{sec-lower}
We sketch the derivation of an $\Omega(\min\{n,\frac{k}{n}\varepsilon^{-2}\log\frac{1}{Q}\})$ lower bound for the \emph{expected} number of comparisons. The lower bound is based on a sampling lemma (Corollary~\ref{cor-actual-sampling} in Section~\ref{ap-lower}) about elements with a certain rank among all sampled elements.
We assume that $\bm{4n\varepsilon\leq k}$.
If $k\leq n\varepsilon$, the $\Omega(\varepsilon^{-1}\log\frac{1}{Q})$ lower bound for the approximate minimum selection problem~\cite{LeucciL20} applies,
and if $n\varepsilon<k<4n\varepsilon$, we multiply the value of $\varepsilon$ by 4 so that $k\leq n\varepsilon$ and the former argument still works,
which does not change the lower bound asymptotically.

Let $T$ be the decision tree of any \emph{randomized} algorithm that solves $\FT(k,\varepsilon)$ with probability at least $1-Q$. 
$T$ is said to $\emph{look at}$ an element $x$ if $T$ performs at least one comparison involving $x$.
Let $\TD$ be the \emph{expected} number of elements that $T$ looks at.
Since $\TD$ is not larger than twice the expected number of comparisons,
it is sufficient to lower bound $\TD$. 
We assume that there are \textbf{no comparison faults}, which does not increase the lower bound and is easier for analysis.

If $\TD\geq \frac{n}{10}$, then $\TD=\Omega(n)$.
Below, we deal with the case that $\TD<\frac{n}{10}$. 
Markov's inequality implies that $T$ looks at more than $2\TD$ elements with probability at most $\frac{1}{2}$.
We construct a new decision tree $\tT$ based on $T$:
$\tT$ first simulates $T$ until reaching a leaf $u$ of $T$ that returns an element $x$, and then conducts three additional steps \emph{sequentially}:
\begin{enumerate}[label=(\alph*)]
	\item If $T$ does not look at $x$, then $\tT$ compares $x$ with another element.
	\item If $\tT$ has looked at fewer than $2\TD + \ceil{\frac{8 n}{k}}$ elements so far, then $\tT$ performs more comparisons such that $\tT$ has looked at \emph{exactly} $2\TD + \ceil{\frac{8 n}{k}}$ elements after this step.
	\item $\tT$ compares all pairs of elements that it has looked at, and then returns $x$.
\end{enumerate}

Intuitively, $\tT$ represents the same algorithm as $T$,
but these additional steps give $\tT$ the following nice properties for analysis (as shown in Lemma~\ref{lem-property-tree}, these properties follow directly from the three additional steps above):
\newcommand{\lempropertytree}{
	\begin{enumerate}[label=(\arabic*)]
		\item $\tT$ knows the sorted order of the elements that $\tT$ has looked at.
		\item $\tT$ has success probability at least $1-Q$.
		\item $\tT$ looks at exactly $2\TD+\ceil{\frac{8 n}{k}}$ elements with probability at least $1/2$.
		Note that this includes the elements that $\tT$ looks at during its simulation of $T$.
	\end{enumerate}}
\lempropertytree

Let us consider the execution of $\tT$ on a uniformly shuffled input.
By property (1), the element returned by a fixed leaf of $\tT$
will always have the same rank among the elements $\tT$ has looked at,
independent of order of the input. (Note that we assumed there are no comparison faults.)
By property (3), with probability at least $1/2$,
the execution of $\tT$ reaches a leaf after looking at exactly $2\TD+\ceil{\frac{8 n}{k}}$
elements. By applying a sampling lemma (Corollary~\ref{cor-actual-sampling})
to each such leaf, we can lower bound the failure probability of $\tT$.

\newcommand{\lemprobabilitytree}{If $k \geq 200$ and $4 n \varepsilon < k$, then the failure probability of $\tT$ on a uniformly shuffled input is at least
	\[\frac{1}{2}\cdot\LC\cdot e^{-24\varepsilon^2\frac{n}{k}(2\TD + \ceil{\frac{8 n}{k}})}\mbox{\quad for a constant }\eta.\]}

\begin{lemma}\label{lem-probability-tree}
	\lemprobabilitytree
\end{lemma}

Since $\tT$ succeeds with probability at least $1-Q$, we have
$Q\geq \frac{1}{2}\LC\cdot e^{-24\varepsilon^2\frac{n}{k}(2\TD+\ceil{\frac{8 n}{k}})}$, implying that $\TD=\Omega(\frac{k}{n}\varepsilon^{-2}\log\frac{1}{Q})$.
We can conclude the following main theorem.

\newcommand{\thmlower}{If $Q<\frac{1}{2}$, then the expected number of comparisons performed by any randomized algorithm that solves the $\FT(k,\varepsilon)$ problem with probability at least $1-Q$ is $\Omega\big(\min\big\{n,\frac{k}{n}\varepsilon^{-2}\log\frac{1}{Q}\big\}\big)$.}

\begin{theorem}\label{thm-lower}
	\thmlower
\end{theorem}

\subsection{Sampling lemma}\label{sub-sampling}

In the previous part, we reduced a general algorithm to returning
an element of a certain rank among all elements the algorithm looked at.
We will now derive a sampling lemma for this case.
For ease of exposition, we also use $\beta$ to denote $\frac{k}{n}$.

\begin{lemma}\label{lem-actual-sampling}
	Let $A$ consist of $m\leq \frac{n}{4}$ elements sampled from $S$ without replacement. 
	Suppose that $m \beta\geq 8$ and that $\frac{1}{2} \geq \beta \geq 4 \varepsilon$. 
	Then there is an absolute constant $\eta$ with the following properties.
	\begin{enumerate}
		\item Let $u$ be the $r$-th smallest element of $A$.
		If $r\leq \ceil{\beta m}$, then $u$ is small with probability at least
		\[\LC\cdot e^{-12\frac{\varepsilon^2}{\beta(1-\beta)}m}.\]
		\item Let $v$ be the $r$-th largest element of A. 
		If $r\leq \ceil{(1-\beta)m}$, then $v$ is large with probability at least
		\[\LC\cdot e^{-12\frac{\varepsilon^2}{\beta(1-\beta)}m}.\]
	\end{enumerate}
\end{lemma}

As every element is either among the $\ceil{\beta m}$ smallest
or among the $\ceil{(1-\beta)m}$ largest ones, the lemma directly implies the following.

\newcommand{\corsampling}{
	Let $A$ consist of $m\leq \frac{n}{4}$ elements sampled from $S$ \emph{without} replacement. 
	Suppose that $m \beta \geq 8$ and that $\frac{1}{2} \geq \beta \geq 4 \varepsilon$.
	Then, an arbitrary element $u$ in $A$ is NOT relevant with probability at least
	\[\LC\cdot e^{-24\cdot \frac{n}{k}\cdot\varepsilon^2\cdot m}.\]
	for some absolute constant $\LC$.}
\begin{corollary}\label{cor-actual-sampling}
\corsampling
\end{corollary}

Now let us briefly sketch the proof of Lemma~\ref{lem-actual-sampling}.
The main observation is that the number of small (or large) elements
in $A$ has hypergeometric distribution. The probability density function
of the hypergeometric distribution can be expressed explicitly
with binomial coefficients. By the entropy bound
for binomial coefficients and a second order tangent bound based on
the second derivative in $x = \frac{\m}{\hM}$,
the following theorem follows.

\begin{theorem}
	Let $X \sim \Hyper(\hM, \hK, \m)$. Let $0 \leq \ell \leq \m$ be an integer with
	$\ell < \hK$ and $\m-\ell < \hM-\hK$. Put $a = \frac{k}{l}$, $b = \frac{\hK}{\hM}$, and $x = \frac{\m}{\hM}$, then we have we have
	$$
	\Prpr{X = \ell} \geq \sqrt{\frac{\pi}{32}} \sqrt{\frac{1 - x}{\m a (1-a)}} e^{-F}
	$$
	for
	$$
	F = \pa{\Div{a}{b} + \frac{x}{1-x} \cdot \frac{(a-b)^2}{2 (b- a x)((1-b)-(1-a) x)}} \cdot \m
	$$
	where $\Div{a}{b}$
	is the Kullback-Leibler divergence (Definition~\ref{def-KLdiv}).
\end{theorem}

By summing this over the tail, we get the following tail
bound, from which Lemma~\ref{lem-actual-sampling} follows.

\newcommand{\thmfirsttail}{
	Let $X\sim \Hyper(\hM,\hK, \m)$. 
	Let $0\leq\ell\leq \m$ be a real number with $\ell < \hK$ and $\m-\ell < \hM-\hK$. 
	Put $a = \frac{\ell}{\m}$, $b = \frac{\hK}{\hM}$ and $x = \frac{\m}{\hM}$. If $a\leq \frac{8}{5}b$, $(1-a)\leq 2(1-b)$ , $x\leq \frac{1}{4}$ and $\m a (1-a) \geq 4$, then we have
	\[Pr[X\geq \ell]\geq \sqrt{\frac{\pi}{320}}\cdot e^{-24}\cdot e^{-\frac{6(a-b)^2}{b(1-b)}\m}.\]}
\begin{corollary}
    \thmfirsttail
\end{corollary}

For a detailed derivation of these bounds, see Appendix~\ref{ap-lower}.

\newpage

\newpage

\appendix

\noindent\huge{\textbf{Appendix}}\normalsize

\section{Supplementary material for Section~\ref{sec-pre}}\label{ap-pre}

\begin{lemma}[Chernoff Bound]\label{lem-chernoff}
Let $X$ be the sum of independent Bernoulli random variables.
If $A\leq E[X]\leq B$, then for any $\delta\in(0,1)$,
\[\Pr[X\geq (1+\delta)\cdot B]\leq e^{-\frac{\delta^2}{3}B} \mbox{\qquad and \qquad} \Pr[X\leq (1-\delta)\cdot A]\leq e^{-\frac{\delta^2}{2}A}.\]
\end{lemma}
\begin{proof}
The two statements can be extended from the proofs of \cite[Theorem~4.4(2)]{MitzenmacherU17} and \cite[Theorem~4.5(2)]{MitzenmacherU17}, respectively. 
Here, we only state the difference.
Since $X$ is the sum of \emph{independent} Bernoulli random variables,
by \cite[Section~4.2.1]{MitzenmacherU17}
\[E[e^{tX}]\leq e^{(e^t-1)E[X]}.\]

For the first claim, using any $\bm{t>0}$,
\[\Pr[X \geq (1+\delta) \cdot B] = \Pr[e^{tX} \geq e^{t(1+\delta) \cdot B}] \leq \dfrac{E[e^{tX}]}{e^{t(1+\delta)B}} \leq \dfrac{e^{(e^t - 1)E[X]}}{e^{t(1+\delta)B}} \overset{E[X]\leq B}{\leq} \dfrac{e^{(e^t - 1)B}}{e^{t(1+\delta)B}}.\]
The remaining steps are identical to the proof of \cite[Theorem~4.4(2)]{MitzenmacherU17}.

For the second claim, using any $\bm{t<0}$,
\[\Pr[X \leq (1-\delta) \cdot A] = \Pr[e^{tX} \geq e^{t(1-\delta) \cdot A}] \leq \dfrac{E[e^{tX}]}{e^{t(1-\delta)A}} \leq \dfrac{e^{(e^t - 1)E[X]}}{e^{t(1-\delta)A}} \overset{A\leq E[X]}{\leq} \dfrac{e^{(e^t - 1)A}}{e^{t(1+\delta)A}}.\]
The remaining steps are identical to the proof of \cite[Theorem~4.5(2)]{MitzenmacherU17}.
\end{proof}

\qquad\\

\rephrase{Lemma}{\ref{lem-majority-vote}}{\lemmajorityvote}
\begin{proof}
Let $\{X_i\mid 1\leq i\leq 2c_p \cdot t + 1\}$ be $2c_p \cdot t + 1$ \emph{independent} Bernoulli random variables such that $X_i=1$ if the $i$-th comparison succeeds, i.e., $\Pr[X_i=1]=1-p$ and $\Pr[X_i=0]=p$.
Let $X = \sum_{i=1}^{2c_p \cdot t + 1} X_i$. 
Then, $E[X]=(2c_p \cdot t + 1)(1-p)$. 
Since $p < \frac{1}{2}$, we know $2(1-p) > 1$ and we can apply Lemma~\ref{lem-chernoff} to prove the first statement as follows:
\begin{alignat*}{2}
 \Pr[X \leq \frac{2c_p \cdot t + 1}{2}]& = \Pr[X \leq \frac{1}{2(1-p)}E[X]]=\Pr[X \leq \left(1-\frac{1-2p}{2-2p}\right)E[X]]  \\
& \underbrace{\leq \exp\left(-\frac{1}{2} \cdot (\frac{1-2p}{2-2p})^2 \cdot E[X]\right)}_{\mbox{Lemma~\ref{lem-chernoff}}}  \\
& = \exp\left(-\frac{1}{2} \cdot (\frac{1-2p}{2-2p})^2 \cdot (2c_p \cdot t + 1)(1-p)\right)\\
& =\exp\left((2c_p\cdot t+1)\frac{(1-2p)^2}{8(1-p)}\right)< \exp\left(-c_p t \frac{(1-2p)^2}{4(1-p)}\right).
\end{alignat*}
which satisfies the statement if we choose $c_p = \lceil \frac{4(1-p)}{(1-2p)^2} \rceil$.
Since $X$ is a binomial random variable and $c_p$ is an integer, the second statement comes as follows:
\[\Pr[X \leq \frac{2c_p \cdot t + 1}{2}]=\Pr[X \leq c_p \cdot t]=\sum_{i=0}^{\MV_p\cdot t} {2\MV_p\cdot t+1\choose i}(1-p)^i p^{2\MV_p\cdot t+1-i}.\]
\end{proof}

\section{Supplementary material for Section~\ref{sec-reduction}}\label{ap-reduction}

\begin{lemma}\label{lem-log-Q-sample}
Let $m=2^{10}\cdot 3^2 \cdot \ln \frac{2}{Q}$, let $X_1,X_2\ldots, X_m$ be $m$ identically and independently distributed Bernoulli random variables with probability $p\geq \frac{8}{9}$, and let $X=\sum_{i=1}^mX_i$.
\[\Pr[X\geq \frac{7}{8}m]\geq 1-\frac{Q}{2}.\]   
\end{lemma}
\begin{proof}
It is sufficient to prove that $\Pr[X \leq \frac{7}{8}m]\leq \frac{Q}{2}$.
Since $p \geq \frac{8}{9}$, $E[X] \geq \frac{8}{9}m$. 
By Lemma~\ref{lem-chernoff},
\begin{align*}
\Pr[X \leq \frac{7}{8}m] & = \Pr[X \leq (1- \frac{1}{64})\cdot \frac{8}{9}m] \overset{\mbox{Lemma~\ref{lem-chernoff}}}{\leq} \exp\left(-\frac{1}{2} \cdot (\frac{1}{64})^2\cdot \frac{8}{9}m\right)\\
&  = \exp\left(-\frac{1}{2^{10}\cdot 3^2}m\right) \leq \exp\left(-\frac{2^{10}\cdot 3^2 \cdot \ln \frac{2}{Q}}{2^{10}\cdot 3^2}\right) = e^{-\ln \frac{2}{Q}}=\frac{Q}{2}.
\end{align*}
\end{proof}

\qquad\\

\rephrase{Theorem}{\ref{thm-expected}}{\thmexpected}
\begin{proof}
Let $m= 2^{10}\cdot 3^2 \cdot \ln \frac{2}{Q}$ as in Lemma~\ref{lem-log-Q-sample}.
The algorithm consists of two stages.
The first stage aims to select $m$ elements in which all elements in the range $(\frac{1}{8}m,\frac{7}{8}m]$ are relevant, and the second stage aims to select an element from $(\frac{1}{8}m,\frac{7}{8}m]$.

For the number of comparisons,
by Theorem~\ref{thm-selection-constant},
it takes $O(\frac{k}{n}\varepsilon^{-2})$ comparisons to select a relevant element with probability at least $1-\frac{1}{9}$, 
so the first stage takes $O(\frac{k}{n}\varepsilon^{-2}\cdot m)=O(\frac{k}{n}\varepsilon^{-2}\cdot\log \frac{1}{Q})$ comparisons.
For the second stage, by Section~\ref{sec-reduction}, one verification step performs $O(\log \frac{1}{Q})$ comparisons.
To derive the expected total number of comparisons, we need to calculate the probability of conducting the $i$-th round.
Since the probability of picking an element in $(\frac{2}{8}m,\frac{6}{8}m]$ is $\frac{1}{2}$ at any round and such an element is verified in $(\frac{1}{8}m,\frac{7}{8}m]$ with probability at least $1-\frac{Q}{2}\geq \frac{1}{2}$ at any round, any round returns an element with probability at least $\frac{1}{2}\cdot \frac{1}{2}=\frac{1}{4}$.
Similar to geometric distribution, the probability that the $i$-th round is conducted is at most $(1-\frac{1}{4})^{i-1}=(\frac{3}{4})^{i-1}$, so the second stage takes expected $\sum_{i=1}^{\infty}(\frac{3}{4})^{i-1}O(\log\frac{1}{Q})=O\big(\log\frac{1}{Q}\cdot \sum_{i=1}^{\infty}(\frac{3}{4})^{i-1}\big)=O(\log\frac{1}{Q})$ comparisons.
To sum up, the algorithm takes expected $O(\frac{k}{n}\varepsilon^{-2}\cdot \log \frac{1}{Q})$ comparisons.

For the success probability,
by Theorem~\ref{thm-selection-constant} and Lemma~\ref{lem-log-Q-sample}, the first stage fails with probability at most $\frac{Q}{2}$. 
The second stage fails only when returning an element in $ [1,\frac{1}{8}m]$ or $(\frac{7}{8}m,m]$.
Since a single round picks an element in  $ [1,\frac{1}{8}m]\cup (\frac{7}{8}m,m]$ with probability $\frac{1}{4}$ and the verification fails with probability at most $\frac{Q}{2}$.
a single round returns an element in $ [1,\frac{1}{8}m]\cup (\frac{7}{8}m,m]$ with probability \emph{at most} $\frac{1}{4}\cdot \frac{Q}{2}=\frac{Q}{8}$, 
Therefore, the failure probability of the second stage is at most $\sum_{i\geq 1}(\frac{3}{4})^{i-1}\cdot \frac{Q}{8}= \frac{Q}{2}$, concluding the following theorem:
\end{proof}

\section{Supplementary material for Section~\ref{sec-median}}\label{ap-median}

\rephrase{Lemma}{\ref{lem-symmetry-median}}{\lemsymmetrymedian}
\begin{proof}
Let $q$ be the failure probability of one majority vote.
Since one majority vote consists of $2c_p\cdot 4+1$ comparisons, by Lemma~\ref{lem-majority-vote}, 
$q\leq e^{-4}$.
If all three majority votes succeed, then the algorithm will return the median, implying that the algorithm will return the median with probability at least $(1-q)^3\geq 1-3q\geq 1-3\cdot e^{-4}\geq 1-\frac{1}{13}$.

Now, we will prove that the algorithm returns the minimum and the maximum with the same probability.
Since there are three majority votes, there are 8 possibilities, and these 8 possibilities lead to four different situations:
exactly the minimum or exactly the median or exactly the maximum gets one point, or all the three elements get one point.
A tree diagram for these 8 possibilities can easily calculate the probabilities of the four situations.
In detail, exactly the minimum (resp. exactly the maximum) gets one point with probability $q(1-q)$, exactly the median gets one point with probability $(1-q)^3+q^3$, and all three elements get one point with probability $q(1-q)$.
Since the algorithm returns an element uniformly at random when all the three elements get one point, the algorithm returns the minimum and the maximum with the same probability $\frac{4}{3}q(1-q)$.

Since the algorithm returns the median with probability at least $1-\frac{1}{13}$ and returns the minimum and the maximum with the same probability, the probability that the algorithm returns the minimum (resp. the maximum) is at most $\frac{1}{26}$.
\end{proof}

\quad

\rephrase{Lemma}{\ref{lem-selected-median-three}}{\lemselectedmedianthree}
\begin{proof}
We only prove the small case, and it is symmetric to the large case.
Let $p_s$ denote the probability that an element randomly picked from $M_{i-1}$ is small.
Since $M_{i-1}$ is good, all elements in its range $(\frac{1}{2}n_{i-1} - n_{i-1} \varepsilon_{i-1}, \frac{1}{2}n_{i-1} + n_{i-1} \varepsilon_{i-1}]$ are relevant, and $p_s\leq \frac{1}{2}-\varepsilon_{i-1}$.
Let $p_1$, $p_2$ and $p_3$ denote the probabilities that the median selection algorithm in Lemma~\ref{lem-symmetry-median} returns the minimum, the median and the maximum of three elements, respectively.
By Lemma~\ref{lem-symmetry-median}, 
$p_2 \geq \frac{12}{13}$, and $p_1=p_3\leq \frac{1}{26}$. 
Also recall that $\varepsilon_{i-1}\leq \frac{1}{6}$.
Then, the probability that an element in $M_i$ is small is

\begin{align*}
& \underbrace{p_s^3}_{\mbox{three small}} + \underbrace{3p_s^2(1-p_s)}_{\mbox{two small \& one non-small}}\cdot (1-p_3)+\underbrace{3p_s(1-p_s)^2}_{\mbox{one small \& two non-small}}\cdot p_1\\
& = p_s^3 + 3p_s^2(1-p_s)(1-p_1)+3p_s(1-p_s)^2\cdot p_1\\
& = \left(3p_s^2-2p_s^3\right)+p_1\cdot\underbrace{\left(3p_s-9p_s^2+6p_s^3\right)}_{\geq 0 \mbox{ since } 0\leq p_s\leq \frac{1}{2}}\\
& \overset{p_1\leq \frac{1}{26}}{\leq} \frac{1}{26}\cdot \underbrace{\left(3p_s+69p_s^2-46p_s^3\right)}_{f(x):=-46x^3+69x^2+3x\mbox{ \& } f'(x)>0 \mbox{ for } 0\leq x\leq 1}\\
& \overset{p_s\leq \frac{1}{2}-\varepsilon_{i-1}}{\leq}\frac{1}{26}\cdot\left(3\left(\frac{1}{2}-\varepsilon_{i-1}\right)+69\left(\frac{1}{2}-\varepsilon_{i-1}\right)^2-46\left(\frac{1}{2}-\varepsilon_{i-1}\right)^3\right)\\
& = \frac{1}{2}-\frac{75}{52}\varepsilon_{i-1}+\frac{23}{13}\varepsilon_{i-1}^3\\
& \overset{\varepsilon_{i-1}<\frac{1}{6}}{\leq} \frac{1}{2}-\frac{75}{52}\varepsilon_{i-1}+ \frac{23}{468}\varepsilon_{i-1}\\
& = \frac{1}{2}-\frac{652}{468}\varepsilon_{i-1}\\
& \leq  \frac{1}{2}-\frac{4}{3}\varepsilon_{i-1}
\end{align*}
\end{proof}

\quad

\rephrase{Lemma}{\ref{lem-good-median}}{\lemgoodmedian}
\begin{proof}
Assume that $M_{i-1}$ is good.
Let $X_i$ be the number of small elements in $M_i$
and let $Y_i$ be the number of large elements in $M_i$.
For the statement, it is sufficient to prove that $\Pr[X_i\geq \frac{n_i}{2}-n_i\varepsilon_i]\leq e^{-5i}$ and $\Pr[Y_i\geq \frac{n_i}{2}-n_i\varepsilon_i]\leq e^{-5i}$.
We will prove the first claim, and it is symmetric to the second claim.
By Lemma~\ref{lem-selected-median-three}, 
\[E[X_i] \leq (\frac{1}{2} - \frac{4}{3}\varepsilon_{i-1})n_i = (\frac{1}{2} - \frac{4}{3}\cdot \frac{4}{5}\varepsilon_i)n_i =(\frac{1}{2} - \frac{16}{15}\varepsilon_i)n_i=\frac{15-32\varepsilon_i}{30}n_i.\] 
By Lemma~\ref{lem-chernoff} (Chernoff Bound), 
we can get

\begin{align*}
\Pr\left[X_i \geq \left(\frac{1}{2} - \varepsilon_i\right) n_i\right] & = \Pr\left[\left(1 + \frac{\frac{1}{15}\varepsilon_i}{\frac{1}{2} - \frac{16}{15}\varepsilon_i}\right)\left(\frac{1}{2} - \frac{16}{15}\varepsilon_i\right) n_i\right]\\
& =\Pr\left[\left(1 + \underbrace{\frac{2\varepsilon_i}{15 - 32\varepsilon_i}}_{:=\delta}\right)
\underbrace{\left(\frac{15-32\varepsilon_i}{30}n_i\right)}_{\geq E[X_i]}\right]\\
 & \underbrace{\leq \exp\left(-\frac{1}{3} \left(\frac{2\varepsilon_i}{15 - 32\varepsilon_i}\right)^2 \cdot \left(\frac{15-32\varepsilon_i}{30}n_i\right)\right)}_{\mbox{Lemma~\ref{lem-chernoff}}}\\
 & = \exp\left(-\frac{4}{90} \cdot \frac{\varepsilon_i^2}{15-32\varepsilon_i} \cdot n_i\right) \leq  \exp\left(-\frac{4}{90} \cdot \frac{\varepsilon_i^2}{15} \cdot n_i\right)\\
 & = \exp\left(-\frac{2}{675} \cdot \left(\left(\frac{5}{4}\right)^{2i}\varepsilon^2\right) \cdot \left(2000 \cdot i \cdot \left(\frac{4}{5}\right)^{2i} \cdot \varepsilon^{-2}\right)\right)\\
 & \leq e^{-5i}.
\end{align*}
\end{proof}

\section{Supplementary material for Section~\ref{sec-selection}}\label{ap-selection}

\rephrase{Lemma}{\ref{lem-selection-invariant}}{\lemselectioninvariant}
\begin{proof}
We prove by induction.
For $i=0$, by assumption in the first paragraph of Section~\ref{sec-selection}, $\beta>2\varepsilon$, i.e., $\beta_0=\beta>2\varepsilon=2\varepsilon_0$.
Also, $\beta_0=\beta\leq 2^0\cdot\beta$.
Assume that for $i=k\geq 0$, $\beta_k>2\varepsilon_k$ and $\beta_k\leq 2^k\cdot \beta$. Note that $k<L$; otherwise, the $(k+1)$-th round does not exist.
By Section~\ref{sec-selection},
\[\beta_{k+1}= \left(2\beta_{k}-\beta_{k}^2-\varepsilon_{k}^2\right)-2q\left(\beta_{k}-\beta_{k}^2-\varepsilon_{k}^2\right).\]
We first prove that $\beta_{k+1}> 2\varepsilon_{k+1}$ as follows:
\begin{align*}
\beta_{k+1}&= \left(2\beta_{k}-\beta_{k}^2-\varepsilon_{k}^2\right)-2q\left(\beta_{k}-\beta_{k}^2-\varepsilon_{k}^2\right) \overset{q< \frac{1}{20}}{>} \frac{19}{10}\beta_k-\frac{9}{10}(\beta_k^2+\varepsilon_k^2)\\
& = \frac{9}{10}\left(-\left(\beta_k-\frac{19}{18}\right)^2+\left(
\frac{19}{18}\right)^2-\varepsilon_k^{2}\right)\\
& \overset{\beta_k<\frac{1}{8} \mbox{ \& }2\varepsilon_k<\beta_k}{>} \frac{9}{10}\left(-\left(2\varepsilon_k-\frac{19}{18}\right)^2+\left(
\frac{19}{18}\right)^2-\varepsilon_k^{2}\right) \\
 & = \frac{19}{5}\varepsilon_k-\frac{9}{2}\varepsilon_k^2 \overset{\varepsilon_k<\frac{1}{2}\beta_k<\frac{1}{16}}{>} \frac{563}{160}\varepsilon_k > 3\varepsilon_{k}=2\varepsilon_{k+1}.
\end{align*}
Then, we prove that $\beta_{k+1}\leq 2^{k+1}\cdot\beta$ as follows:
\begin{align*}
\beta_{k+1}&= \left(2\beta_{k}-\beta_{k}^2-\varepsilon_{k}^2\right)-2q\left(\beta_{k}-\beta_{k}^2-\varepsilon_{k}^2\right) \overset{q\geq 0}{\leq} 2\beta_{k}-\beta_{k}^2-\varepsilon_{k}^2 \\
&\leq 2\beta_k \leq 2\cdot 2^k\cdot \beta=2^{k+1}\cdot \beta.
\end{align*}
\end{proof}

\quad

\rephrase{Lemma}{\ref{lem-good-selection}}{\lemgoodselection}
\begin{proof}

Assume that $S_{i-1}$ is good.
Let $X_i$ be the number of small elements in $S_i$ and let $Y_i$ be the number of small and relevant elements in $S_i$. 
For the statement, it is sufficient to prove $\Pr[X_i \geq k_i - n_i \varepsilon_i] \leq e^{-4i}$ and $\Pr[Y_i \leq k_i + n_i \varepsilon_i] \leq e^{-4i}$.

Since $S_{i-1}$ is good, all elements in the range $(k_{i-1}-n_{i-1}\varepsilon_{i-1}, k_{i-1}+n_{i-1}\varepsilon_{i-1}]$ of $S_{i-1}$ relevant.
Recall that $\beta_i=\frac{k_i}{n_i}$.
Therefore, according to the way of selecting elements for $S_i$ in Section~\ref{sec-selection},
the probability that an element in $S_i$ is small is \emph{at most}
\begin{align*}
&(\beta_{i-1} - \varepsilon_{i-1})^2 + (1-q) \cdot 2 \cdot (\beta_{i-1} - \varepsilon_{i-1})(1- (\beta_{i-1} - \varepsilon_{i-1}))\\
&=\left(2(\beta_{i-1}-\varepsilon_{i-1})-(\beta_{i-1}-\varepsilon_{i-1})^2\right)-q\left(2(\beta_{i-1}-\varepsilon_{i-1})-2(\beta_{i-1}-\varepsilon_{i-1})^2\right)
\end{align*} 
Let $\bm{p_s}$ denote the above upper bound.
Similarly, the probability that a selected element is small or relevant is \emph{at least}
\begin{align*} 
&(\beta_{i-1} + \varepsilon_{i-1})^2 + (1-q) \cdot 2 \cdot (\beta_{i-1} + \varepsilon_{i-1})(1-(\beta_{i-1} + \varepsilon_{i-1}))\\
&=\left(2(\beta_{i-1}+\varepsilon_{i-1})-(\beta_{i-1}+\varepsilon_{i-1})^2\right)-q\left(2(\beta_{i-1}+\varepsilon_{i-1})-2(\beta_{i-1}+\varepsilon_{i-1})^2\right)
\end{align*}
Let $\bm{p_{sr}}$ denote the above lower bound.
Then, 
\[E[X_i] \leq p_s \cdot n_i \mbox{\qquad and\qquad} E[Y_i] \geq p_{sr} \cdot n_i\]. 

By the formulation of $\beta_i$ in Section~\ref{sec-selection}, we can re-formulate $\beta_i$ with $p_s$ and $p_{sr}$ as follows:
\[\beta_i=\frac{p_s+p_{sr}}{2}.\]
Therefore, we can reformulate $\Pr[X_i \geq k_i - n_i \varepsilon_i]$ and $\Pr[Y_i \leq k_i + n_i \varepsilon_i]$ as follows:
\begin{align*} 
\Pr\left[X_i \geq k_i - \varepsilon_i n_i\right] & = \Pr\left[X_i \geq (\beta_i - \varepsilon_i)n_i\right] = \Pr\left[X_i \geq \left(\frac{p_s + p_{sr}}{2} - \frac{3}{2}\varepsilon_{i-1}\right)n_i\right]\\
& = \Pr\left[X_i \geq \left(1 + \frac{p_{sr} - p_s - 3 \varepsilon_{i-1}}{2p_s}\right)\cdot p_s n_i\right]
\end{align*}
and similarly,
\[\Pr\left[Y_i \leq k_i + \varepsilon_i n_i\right] = \Pr\left[Y_i \geq \left(1 - \frac{p_{sr} - p_s - 3 \varepsilon_{i-1}}{2p_{sr}}\right)\cdot p_{sr} n_i\right].\]

In order to apply Lemma~\ref{lem-chernoff} (Chernoff bound),
we need to show that (1) $p_{sr} - p_s - 3 \varepsilon_{i-1}>0$, (2) $p_{sr} - p_s - 3 \varepsilon_{i-1}< 2p_s$ and (3) $p_{sr} - p_s - 3 \varepsilon_{i-1} <  2p_{sr}$.
Since $p_{sr}> p_s$, it is sufficient to prove the first two inequalities.

For the first inequality,
\begin{align*} 
p_{sr} - p_s - 3 \varepsilon_{i-1} &  = (1-4q) \cdot \varepsilon_{i-1} - (1-2q) \cdot 4 \beta_{i-1} \varepsilon_{i-1}\\
& \overset{\beta_{i-1} < \frac{1}{8},\; 1-2q>0}{>} \frac{1-6q}{2}\varepsilon_{i-1} \overset{q < \frac{1}{20}}{\geq} \frac{7}{20}\varepsilon_{i-1}.
\end{align*}

For the second inequality, we upper bound $p_{sr} - p_s - 3 \varepsilon_{i-1}$ and lower bound $2p_s$:
\begin{align*} 
p_{sr} - p_s - 3 \varepsilon_{i-1}  &   = (1-4q) \cdot \varepsilon_{i-1} - (1-2q) \cdot 4 \beta_{i-1} \varepsilon_{i-1}\\
& \overset{\beta_{i-1} \geq 0,\; 1-2q>0}{\leq} 
(1-4q)\varepsilon_{i-1} \overset{q \geq 0 }{\leq} \varepsilon_{i-1},
\end{align*} 
and 
\begin{align*} 
2p_s  &  =2\left(2(\beta_{i-1}-\varepsilon_{i-1})-(\beta_{i-1}-\varepsilon_{i-1})^2\right)-2q\left(2(\beta_{i-1}-\varepsilon_{i-1})-2(\beta_{i-1}-\varepsilon_{i-1})^2\right)\\
&  \overset{q < \frac{1}{20}}{\geq} \frac{19}{5}(\beta_{i-1}-\varepsilon_{i-1})-\frac{9}{5}(\beta_{i-1}-\varepsilon_{i-1})^2\\
&=\frac{1}{5}(\beta_{i-1}-\varepsilon_{i-1})\left(19-9(\beta_{i-1}-\varepsilon_{i-1})\right)\\
&  \overset{\beta_{i-1}-\varepsilon_{i-1}\leq\; \beta_{i-1}\; \leq \frac{1}{8}}{\geq} \frac{143}{40}(\beta_{i-1}-\varepsilon_{i-1})\\ & \overset{\beta_{i-1}> 2\varepsilon_{i-1}\;(\mbox{Lemma}~\ref{lem-selection-invariant})}{\geq} \frac{143}{40}\varepsilon_{i-1},
\end{align*}
implying that $p_{sr} - p_s - 3 \varepsilon_{i-1}< p_s< p_{sr}$.

For applying the Chernoff bound, it is convenient to have a simple lower bound for $p_{sr} - p_s - 3 \varepsilon_{i-1}$ and simple upper bounds for $p_s$ and $p_{sr}$.
Since we already derive that $p_{sr} - p_s - 3 \varepsilon_{i-1}\geq \frac{7}{20}\varepsilon_{i-1}$, we deal with the other two as follow.
\begin{align*} 
p_s &  =\left(2(\beta_{i-1}-\varepsilon_{i-1})-(\beta_{i-1}-\varepsilon_{i-1})^2\right)-q\left(2(\beta_{i-1}-\varepsilon_{i-1})-2(\beta_{i-1}-\varepsilon_{i-1})^2\right)\\
& \overset{q\geq 0}{\leq} 2(\beta_{i-1}-\varepsilon_{i-1})-(\beta_{i-1}-\varepsilon_{i-1})^2\leq 2\beta_{i-1},
\end{align*}
and
\begin{align*} 
p_{sr} &  =\left(2(\beta_{i-1}+\varepsilon_{i-1})-(\beta_{i-1}+\varepsilon_{i-1})^2\right)-q\left(2(\beta_{i-1}+\varepsilon_{i-1})-2(\beta_{i-1}+\varepsilon_{i-1})^2\right)\\
& \overset{q\geq 0}{\leq} 2(\beta_{i-1}+\varepsilon_{i-1})-(\beta_{i-1}+\varepsilon_{i-1})^2\leq 2(\beta_{i-1}+\varepsilon_{i-1})\\
& \overset{2\varepsilon_{i-1}<\beta_{i-1}\;(\mbox{Lemma}~\ref{lem-selection-invariant})}{\leq} 3\beta_{i-1}.
\end{align*}

\begingroup
\allowdisplaybreaks
\begin{align*}
\Pr[X_i \geq k_i - \varepsilon_i n_i] &  =\Pr\left[X_i \geq \left(1 + \underbrace{\frac{p_{sr} - p_s - 3 \varepsilon_{i-1}}{2p_s}}_{:=\delta}\right)\cdot \underbrace{p_s n_i}_{\geq E[X_i]}\right] 
\\ & \underbrace{\leq \exp\left(-\frac{1}{3} \cdot \left(\frac{p_{sr} - p_s - 3 \varepsilon_{i-1}}{2p_s}\right)^2 \cdot p_sn_i\right)}_{\mbox{(Lemma~\ref{lem-chernoff})}}   \\
&=\exp\left(-\frac{1}{3} \cdot \underbrace{\frac{(p_{sr} - p_s - 3 \varepsilon_{i-1})^2}{4p_s}}_{p_{sr} - p_s - 3 \varepsilon_{i-1}\geq \frac{7}{20}\varepsilon_{i-1}\mbox{ \& } p_s \leq 2\beta_{i-1}} \cdot n_i\right)\\
& \leq\exp\left(-\frac{1}{3}\cdot \frac{\left(\frac{7}{20}\varepsilon_{i-1}\right)^2}{8\beta_{i-1}}\cdot n_i \right)\\
& = \exp\left(\underbrace{-\frac{49}{9600}\cdot \left(\varepsilon_{i-1}^2
	\cdot \beta_{i-1}^{-1}\right)\cdot \left(960\cdot i\cdot (\frac{8}{9})^i \cdot  \frac{k}{n}\varepsilon^{-2}\right)}_{\varepsilon_{i-1}=(\frac{3}{2})^{i-1}\varepsilon\mbox{ \& } \beta_{i-1}^{-1}\geq 2^{-(i-1)}\frac{n}{k}\mbox{ (Lemma~\ref{lem-selection-invariant})}}\right)\\
& \leq \exp\left(-\frac{49}{10}\cdot i\cdot  \left(\frac{9}{4}\right)^{i-1}\left(2^{-(i-1)}\right)\left(\frac{8}{9}\right)^i\right)\\
& \leq e^{-4i}
\end{align*}
\endgroup

\begin{align*}
\Pr[Y_i \leq k_i + \varepsilon_i n_i] & =\Pr\left[Y_i \geq \left(1 + \underbrace{\frac{p_{sr} - p_s - 3 \varepsilon_{i-1}}{2p_{sr}}}_{:=\delta}\right)\cdot \underbrace{p_{sr} n_i}_{\leq E[Y_i]}\right] 
\\ & \underbrace{\leq \exp\left(-\frac{1}{2} \cdot \left(\frac{p_{sr} - p_s - 3 \varepsilon_{i-1}}{2p_{sr}}\right)^2 \cdot p_{sr}n_i\right)}_{\mbox{(Lemma~\ref{lem-chernoff})}}   \\
&=\exp\left(-\frac{1}{2} \cdot \underbrace{\frac{(p_{sr} - p_s - 3 \varepsilon_{i-1})^2}{4p_{sr}}}_{p_{sr} - p_s - 3 \varepsilon_{i-1}\geq \frac{7}{20}\varepsilon_{i-1}\mbox{ \& } p_{sr} \leq 3\beta_{i-1}} \cdot n_i\right)\\
& \leq\exp\left(-\frac{1}{2}\cdot\frac{\left(\frac{7}{20}\varepsilon_{i-1}\right)^2}{12\beta_{i-1}}\cdot n_i \right)\\
& = \exp\left(\underbrace{-\frac{49}{9600}\cdot \left(\varepsilon_{i-1}^2
	\cdot \beta_{i-1}^{-1}\right)\cdot \left(960\cdot i\cdot (\frac{8}{9})^i \cdot  \frac{k}{n}\varepsilon^{-2}\right)}_{\varepsilon_{i-1}=(\frac{3}{2})^{i-1}\varepsilon\mbox{ \& } \beta_{i-1}^{-1}\geq 2^{-(i-1)}\frac{n}{k}\mbox{ (Lemma~\ref{lem-selection-invariant})}}\right)\\
& \leq \exp\left(-\frac{49}{10}\cdot i\cdot  \left(\frac{9}{4}\right)^{i-1}\left(2^{-(i-1)}\right)\left(\frac{8}{9}\right)^i\right)\\
& \leq e^{-4i}
\end{align*}
\end{proof}

\section{Supplementary material for Section~\ref{sec-lower}}\label{ap-lower}

\subsection{Derivations Towards Theorem~\ref{thm-lower}}

This subsection shows detailed proofs for the lower bound analysis in Section~\ref{sec-lower}. 
We first prove a number of nice properties for the auxiliary decision tree $\tT$ (Lemma~\ref{lem-property-tree}).
Then, we use Lemma~\ref{lem-property-tree} and a sampling lemma (Corollary~\ref{cor-sampling} in Appendix~\ref{sub-sampling}) to bound the failure probability of $\tT$ (Lemma~\ref{lem-probability-tree}).
Finally, we combine Lemma~\ref{lem-property-tree} and Lemma~\ref{lem-probability-tree} to prove the lower bound (Theorem~\ref{thm-lower}).

\begin{lemma}\label{lem-property-tree}
	$\tT$ has the following properties:
	\begin{enumerate}
		\item $\tT$ knows the sorted order of the elements that $\tT$ has looked at.
		\item $\tT$ has success probability at least $1-Q$.
		\item $\tT$ looks at exactly $2\TD+\ceil{\frac{8 n}{k}}$ elements with probability at least $1/2$.
		Note that this includes the elements that $\tT$ looks at during its simulation of $T$.
	\end{enumerate}
\end{lemma}
\begin{proof}
	Property~(1) comes from step~(c) in which $\tT$ compares all pairs of elements that $\tT$ has looked at. Remember that we assume no comparison faults for the lower bound analysis.
	
	For property~(2),
	note that $\tT$ first simulates $T$, then does some additional comparison and then returns the element that $T$ would have returned (independent of the outcome of the additional comparisons). Hence $\tT$ has the same success probability as $T$, which is at least $1-Q$ by assumption.

	For property~(3),
	according to the three steps,
	if $T$ looks at no more than $2\TD$ elements, then $\tT$ will look exactly $2\TD+\ceil{\frac{8 n}{k}}$ elements. 
	Since the probability that $T$ looks at more than $2\TD$ elements is at most $\frac{1}{2}$ (by the definition of $\TD$ and by Markov's inequality), property~(3) follows. 
\end{proof}

\qquad\\

\rephrase{Lemma}{\ref{lem-probability-tree}}{\lemprobabilitytree}
\begin{proof}
	Recall that we build $\tilde{T}$ only when $\TD < \frac{n}{10}$.
	Fix a leaf $w$ of $\tilde{T}$. Suppose that the execution of $\tilde{T}$ reaches $w$.
	Let $x$ be the element that $\tilde{T}$ returns and
	let $A$ be the set of elements that $\tT$ has looked at when the execution reaches $w$.
	
	As $\tilde{T}$ is run on a uniformly shuffled input, the distribution of the set $A$ as a random variable is the same as the distribution of a set of $|A|$ elements sampled from $S$ without replacement. 
	Note that since $\tilde{T}$ has only compared elements
	in $A$, these comparisons do not affect the distribution of $A$ as a random variable. 
	By Lemma~\ref{lem-property-tree}.(1), $x$ always has the same rank in $A$. 
	If $|A| = 2\TD+\ceil{\frac{8 n}{k}}$, then $|A| \leq \frac{n}{4}$, and by assumption, we have $4 n \varepsilon < k$. 
	Note that $\frac{k}{n}(2\TD+\ceil{\frac{8 n}{k}})\geq 8$.
	Hence, Corollary~\ref{cor-sampling} implies that $\tT$ fails with probability
	at least
	\[\LC\cdot e^{-24\cdot \frac{n}{k}\cdot\varepsilon^2\cdot |A|}=\LC\cdot e^{-24\cdot \frac{n}{k}\cdot\varepsilon^2\cdot (2\TD+\ceil{\frac{8 n}{k}})}.\]
	To summarize, if $\tT$ reaches a leaf after looking at exactly
	$2\TD+\ceil{\frac{8 n}{k}}$ elements, then $\tT$ fails with at
	least this probabilty. By Lemma~\ref{lem-property-tree}.(3),
	this happens with probability at least $\frac{1}{2}$, leading to the statement.
\end{proof}

\qquad\\

\rephrase{Theorem}{\ref{thm-lower}}{\thmlower}
\begin{proof}
	As discussed in the beginning of Section~\ref{sec-lower},
	if $k< 4n\varepsilon$, the lower bound $\Omega(\varepsilon^{-1}\log\frac{1}{Q})$ for approximate minimum selection~\cite{LeucciL20} applies. Similarly, if $k \leq 200$, we may
	increase $\varepsilon$ by $\frac{200}{n}$, which changes $\varepsilon$ by at most a constant factor,
	and apply the lower bound for the approximate minimum selection~\cite{LeucciL20}.
	Therefore, it is sufficient to consider the case that $4\varepsilon\leq \frac{k}{n}\leq \frac{1}{2}$
	and $k \geq 200$.
	Recall that $T$ is the decision tree of any randomized algorithm that solves $\FT(k,\varepsilon)$ with probability at least $1-Q$ and $\TD$ is the expected number of elements that $T$ looks at. 
	If $\TD\geq \frac{n}{10}$, a lower bound $\Omega(n)$ follows.
	Otherwise, we build the auxiliary decision tree $\tT$.

	By Lemma~\ref{lem-property-tree}.(2), the success probability of $\tT$ is at least $1-Q$, and by Lemma~\ref{lem-probability-tree}, the failure probability of $\tT$ is at least $\frac{1}{2}\cdot\LC\cdot e^{-24\varepsilon^2\frac{n}{k}(2\TD+\ceil{\frac{8 n}{k}})}$ for a constant $\LC$, implying that	
	\[Q\geq \frac{1}{2}\cdot\LC\cdot e^{-24\varepsilon^2\frac{n}{k}(2\TD+\ceil{\frac{8 n}{k}})},\]
	or equivalently
	\[
	\TD \geq \frac{1}{48} \varepsilon^{-2} \frac{k}{n} \ln \frac{\LC}{2 Q} - \frac{1}{2}\ceil{\frac{8 n}{k}}.
	\]
	
	If $Q \leq \frac{\LC}{1000}$, then the first term $\frac{1}{48} \varepsilon^{-2} \frac{k}{n} \ln \frac{\LC}{2 Q}$ dominates the second term $\frac{1}{2}\ceil{\frac{8 n}{k}}$ as $4 \varepsilon \leq \frac{k}{n}$, and thus $\TD=\Omega(\frac{k}{n}\varepsilon^{-2}\ln\frac{1}{Q})$. ($\LC=\sqrt{\frac{\pi}{320}}\cdot e^{-24}$ as stated in Theorem~\ref{thm-first-tail}.)
	
	It remains to analyze the case that $Q>\frac{\eta}{1000}$, for which we construct an auxiliary algorithm that solves the $\FT(k,\varepsilon)$ problem with probability at least $1 - \frac{\LC}{1000}$.
	We will use $\A$ and $\tA$ to denote the original algorithm and the auxiliary algorithm, respectively.
	Recall that $\A$ solves the $\FT(k,\varepsilon)$ problem with probability at least $1-Q$.
	Select $k'$ such that $A$ outputs a small element with probability at most $\frac{k'}{n} - \frac{1-Q}{2}$ and a large element with probability at most $1 - \frac{k'}{n} - \frac{1-Q}{2}$.
	Thus, by using $\A$ to get sampled elements instead of sampling from the input, the $\FT(k,\varepsilon)$ problem is reduced to the $\FT\big(k',\frac{1-Q}{2}\big)$ problem (with the restriction
	that we may only use sampled elements).
	Motivated by this, let $\tA$ be a modified (fault-free) version of our algorithms (Section~\ref{sec-reduction}--\ref{sec-selection}) for the $\FT(k', \frac{1-Q}{2})$ problem with success probability at least $1 - \frac{\LC}{1000}$ in which each sampling from $S$ is implemented by calling $\A$ on $S$.
	The correctness of $\tA$ relies on the fact that our algorithms only sample elements from $S$ uniformly at random and the corresponding analysis
	only cares about the probability of getting a small / relevant / large element.

	As applying our algorithm to solve the $\FT(k',\frac{1-Q}{2})$ problem with probability $1 - \frac{\LC}{1000}$ would sample $O(\frac{k'}{n} (1-Q)^{-2} \log \frac{1000}{\LC})$ times from $S$,
	$\tA$ invokes $\A$ at most $O(\frac{k'}{n} (1-Q)^{-2} \log \frac{1000}{\LC})$ times and thus performs \emph{expected} $O(\TD\frac{k'}{n} (1-Q)^{-2} \log \frac{1000}{\LC})$ comparisons.
	Since all terms except $\TD$ are bounded from above by a constant, the above bound is can be reformulated as $O(\TD)$.
	On the other hand, we have already proven that the expected number of comparison to solve the $\FT(k,\varepsilon)$ problem with probability at least $1-\frac{1000}{\LC}$ is $\Omega(\frac{k}{n}\varepsilon^{-2}\log \frac{1000}{\LC})=\Omega(\frac{k}{n}\varepsilon^{-2})$.
	Since the first bound $O(\TD)$ is an upper bound for the second bound $\Omega(\frac{k}{n}\varepsilon^{-2})$,
	$\TD=\Omega(\frac{k}{n}\varepsilon^{-2})=\Omega(\frac{k}{n}\varepsilon^{-2}\log \frac{1}{Q})$.  
	Recall that $\log \frac{1}{Q}$ is a constant since $Q\geq \frac{\LC}{1000}$ and $\LC$ is an absolute constant.
	
	To sum up, when $4 \varepsilon \leq \frac{k}{n} \leq \frac{1}{2}$ and $k \geq 200$,the expected number of comparisons required by any algorithm that solve $\FT(k,\varepsilon)$ with probability $1-Q$ is
	\[\Omega\left(\min\{n, \frac{k}{n}\varepsilon^{-2}\log\frac{1}{Q}\}\right).\]
\end{proof}

If $\frac{k}{n}\varepsilon^{-2}\log\frac{1}{Q}=\Omega(n)$,
the lower bound in Theorem~\ref{thm-lower} becomes just $\Omega(n)$.
By reducing it to the exact selection problem, we can show a stronger lower bound in this case.
\begin{theorem}\label{thm-lower-another}
	If $Q < \frac{1}{2}$ and $\frac{k}{n}\varepsilon^{-2}\log\frac{1}{Q}=\Omega(n)$, then the expected number of comparisons performed by any randomized algorithm that solves $\FT(k,\varepsilon)$ with probability at least $1-Q$ is
	\[\Omega\left(\max\left\{n, \varepsilon^{-1}\log\frac{\frac{k+n\varepsilon}{2n\varepsilon}}{Q} \right\}\right).\] 
\end{theorem}
\begin{proof}
	The first term $n$ directly comes from the first term $n$ of Theorem~\ref{thm-lower}.
	Recall that we assume $k\leq \frac{n}{2}$.
	The second term $\varepsilon^{-1}\log\frac{(k+n\varepsilon)/(2n\varepsilon)}{Q}$ can be reduced from the lower bound $\Omega(n\log \frac{k}{Q})$ for the exact $k$-th smallest element selection problem~\cite{FeigeRPU94} as follows.
	Note that as remarked in \cite[Section~1]{FeigeRPU94},
	their bound holds both in expectation and in the worst case.
	
	Assume we attempt to select the $\ell$-th smallest element among $m$ elements. 
	We can duplicate each element $2\cdot n\varepsilon$ times and solve the $\FT\big(k,\varepsilon\big)$ problem where $n=m\cdot n\varepsilon$ and $k=(2n\varepsilon)\cdot\ell-n\varepsilon$.
	This setting implies that
	$m=\varepsilon^{-1}$ and $\ell= (k+n\varepsilon)/(2n\varepsilon)$. 
	Since selecting the $\ell$-th smallest element among $m$ elements with probability at least $1-Q$ requires $\Omega(m\log\frac{\ell}{Q})$ comparisons,
	a lower bound of $\Omega(\varepsilon^{-1}\log\frac{(k+n\varepsilon)/(2n\varepsilon)}{Q})$ follows. 
\end{proof}

\subsection{Sampling Lemma}\label{sub-sampling}

This subsection aims to build up a sampling bound (Corollary~\ref{cor-sampling}) that is the key ingredient to prove Lemma~\ref{lem-probability-tree}.
Corollary~\ref{cor-sampling} roughly states that for a set $A$ of randomly sampled elements (without replacement),
the probability that an element of a certain rank in $A$ is NOT relevant decreases as $e^{-\Omega(\frac{\varepsilon^2}{\beta} |A|)}$.
To prove Corollary~\ref{cor-sampling}, we first derive Lemma~\ref{lem-sampling} that deals with different positions in $A$.
For ease of exposition, we also use $\beta$ to denote $\frac{k}{n}$ in the proofs.
As assumed in the whole paper, $\beta\leq \frac{1}{2}$, 
and as stated in Section~\ref{sec-lower}, it is also sufficient to consider $\beta\geq 4\varepsilon$ since if $\beta<4\varepsilon$, we then can apply the lower bound for the approximate minimum selection~\cite{LeucciL20}.

\begin{lemma}\label{lem-sampling}
	Let $A$ consist of $m\leq \frac{n}{4}$ elements sampled from $S$ without replacement. 
	Suppose that $m \beta\geq 8$ and that $\frac{1}{2} \geq \beta \geq 4 \varepsilon$. 
	Then there is an absolute constant $\eta$ with the following properties. (For the value of $\eta$, see Theorem~\ref{thm-first-tail}.)
	\begin{enumerate}
		\item Let $u$ be the $r$-th smallest element of $A$.
		If $r\leq \ceil{\beta m}$, then $u$ is small with probability at least
		\[\LC\cdot e^{-12\frac{\varepsilon^2}{\beta(1-\beta)}m}.\]
		\item Let $v$ be the $r$-th largest element of A. 
		If $r\leq \ceil{(1-\beta)m}$, then $v$ is large with probability at least
		\[\LC\cdot e^{-12\frac{\varepsilon^2}{\beta(1-\beta)}m}.\]
	\end{enumerate}
\end{lemma}
\begin{proof}
	We first prove (1).
	Let $X$ denote the number of small elements in $A$.
	Then $X\sim \Hyper(n, (\beta-\varepsilon)k,m)$ has a hypergeometric distribution (Definition~\ref{def-hyper} in Appendix~\ref{sub-first-tail}).
	Since $r\leq \ceil{\beta m}$, $u$ is small if and only if $A$ contains at least $r$ small elements, i.e., if and only if $X\geq r$.
	Put $a=\beta$ and b =$\beta-\varepsilon$.
	Then we have $a\leq \frac{8}{5}b$ and $(1-a)\leq \frac{8}{5}(1-b)$
	as $\beta \geq 4 \varepsilon$. As $m \beta \geq 8$ and $\beta \leq \frac{1}{2}$,
	we also have $m a (1-a) \geq 4$.
	Hence by Theorem~\ref{thm-first-tail}
	\[\Pr[X\geq r]\geq \Pr[X\geq \ceil{\beta m}] = \Pr[X\geq \beta m]\geq \LC\cdot e^{-6\frac{\varepsilon^2}{b(1-b)}}\]
	for some absolute constant $\LC$.
	Since $\beta\geq 2\varepsilon$, we have $b\geq\frac{\beta}{2}$,
	and since we also have $(1-b)\geq (1-\beta)$, we have 
	\[\frac{\varepsilon^2}{b(1-b)}\leq \frac{2 \varepsilon^2}{\beta (1-\beta)},\]
	implying that
	\[\Pr[X\geq r]\geq \LC\cdot e^{-12\frac{\varepsilon^2}{\beta (1-\beta)}m}\]
	The proof of (2) is symmetric with large elements instead of small ones
	and with $(1-\beta)$ instead of $\beta$.
\end{proof}

\begin{corollary}\label{cor-sampling}
	Let $A$ consist of $m\leq \frac{n}{4}$ elements sampled from $S$ without replacement. 
	Suppose that $m \beta \geq 8$ and that $\frac{1}{2} \geq \beta \geq 4 \varepsilon$.
	Then, an arbitrary element $u$ in $A$ is NOT relevant with probability at least
	\[\LC\cdot e^{-24\cdot \frac{n}{k}\cdot\varepsilon^2\cdot m}.\]
	for some absolute constant $\LC$. (For the value of $\eta$, see Theorem~\ref{thm-first-tail}.)
\end{corollary}
\begin{proof}
	Let $r$ be the rank of $u$ in $A$.
	If $r\leq \ceil{\beta m}$, then by part (1) of Lemma~\ref{lem-sampling}, $u$ is small with probability at least
	\[\LC\cdot e^{-12\frac{\varepsilon^2}{\beta(1-\beta)}m}.\]
	Otherwise, $r\geq \ceil{\beta m}+1 \geq \beta m + 1$, so $m+1-r\leq (1-\beta)m \leq \ceil{(1-\beta) m}$. 
	Since $u$ is the $(m+1-r)$-th largest element of $A$, by part (2) of Lemma~\ref{lem-sampling}, $u$ is large with probability at least
	\[\LC\cdot e^{-12\frac{\varepsilon^2}{\beta(1-\beta)}m}.\]
	Since $1-\beta \geq \frac{1}{2}$, we have
	\[\LC\cdot e^{-12\frac{\varepsilon^2}{\beta(1-\beta)}m}\geq \LC\cdot e^{-24\frac{\varepsilon^2}{\beta}m}=\LC\cdot e^{-24\cdot \frac{k}{n}\cdot\varepsilon^2\cdot m}.\]
\end{proof}

\subsection{A lower tail for hypergeometric distribution}\label{sub-first-tail}

This subsection aims to build a lower tail bound for the hypergeometric distribution (Theorem~\ref{thm-first-tail}), which is used in the proof of Lemma~\ref{lem-sampling}.
We first define the hypergeometric distribution and the Kullback-Leibler divergence.
Then, we prove Lemma~\ref{lem-L4} for the Kullback-Leibler divergence and derive Corollary~\ref{cor-C2}.
Finally, we adopt Corollary~\ref{cor-C2} to prove Theorem~\ref{thm-first-tail}.

\begin{definition}\label{def-hyper}
	Consider $\hM$ balls, out of which $\hK$ balls are black and $\hM-\hK$ balls are white.
	$\Hyper(\hM,\hK, \m)$ is the probability distribution for the number of black balls in $\m$ draws from the $\hM$ balls using sampling without replacement, which is the so-called \emph{hypergeometric distribution}.
	$X\sim \Hyper(\hM,\hK, \m)$ means that $X$ is a random variable with $\Hyper(\hM,\hK, \m)$ distribution.
\end{definition}

\begin{definition}\label{def-KLdiv}
	For $a, b \in \pa{0, 1}$, the \emph{Kullback-Leibler divergence} $\Div{a}{b}$ is given by
	\[
	\Div{a}{b} := a \ln\pa{\frac{a}{b}} + (1-a) \ln \pa{\frac{1-a}{1-b}}.
	\]
\end{definition}

\begin{lemma}\label{lem-L4}
	Let $a\in (0,1)$, then
	\[D(a\parallel a+y)\;\leq\; 2\frac{y^2}{a(1-a)},\qquad \forall y \in \left[-\frac{a}{2},\quad \frac{1-a}{2}.\right]\]
\end{lemma}
\begin{proof}
	We have
	\begin{align*}
	\Div{a}{b} &= a \ln\pa{\frac{a}{b}} + (1-a) \ln \pa{\frac{1-a}{1-b}}\\
	&= - a \ln \pa{\frac{b}{a}} - (1-a) \ln \pa{\frac{1-b}{1-a}}\\
	&= -a \ln \pa{1 + \frac{b-a}{a}} - (1-a) \ln \pa{1 + \frac{a-b}{1-a}}
	\end{align*}
	Putting $b = a + y$ gives us
	\begin{align*}
	\Div{a}{a+y} &= -a \ln \pa{1 + \frac{y}{a}} - (1-a) \ln \pa{1 - \frac{y}{1-a}}
	\end{align*}
	For $x \geq -\frac{1}{2}$, we have $\ln(1+x) \geq x - x^2$.
	Note that $\frac{y}{a} \geq -0.5$ and $-\frac{y}{1-a} \geq -0.5$ by assumption.
	Therefore
	\begin{align*}
	\Div{a}{a+y} &\leq -a (\frac{y}{a} - \frac{y^2}{a^2}) - (1-a) \pa{-\frac{y}{1-a} - \frac{y^2}{(1-a)^2}}\\
	&= \frac{y^2}{a} + \frac{y^2}{1-a}\\
	&= \frac{y^2}{a (1-a)}
	\end{align*}
\end{proof}

\begin{corollary}\label{cor-C2}
	Let $X\sim \Hyper(\hM,\hK,\m)$. 
	Let $0 < \ell < \m$ be an integer.
	Put $a = \frac{\ell}{\m}$, $b = \frac{\hK}{\hM}$ and $x = \frac{\m}{\hM}$. If $a\leq 2b$, $(1 -a)\leq 2(1 - b)$ and $x\leq \frac{1}{4}$, then we have
	\[Pr[X = \ell]\geq \sqrt{\frac{\pi}{64\m a(1-a)}}\cdot e^{-3\frac{(a-b)^2}{b(1-b)}\m}.\]
\end{corollary}
\begin{proof}
	Note that we have $\ell = a x \hM \leq \frac{2 b}{4} \hM < \hK$
	and $m-\ell = (1-a) x \hM \leq \frac{2(1-b)}{4} \hM < \hM -\hK$.
	By Theorem~\ref{thm-tool}, we have
	\[\Pr[X=\ell]\geq \sqrt{\frac{\pi}{32}}\sqrt{\frac{1-x}{\m a(1-a)}}\cdot e^{-F}\geq \sqrt{\frac{\pi}{64\m a(1-a)}}\cdot e^{-F},\]
	where the last inequality comes from the fact that $x \leq \frac{1}{2}$
	Next, we bound $F$.
	As $x\leq \frac{1}{4}$, we have $\frac{x}{1-x}\leq \frac{1}{2}$.
	As $a\leq 2b$, we have $b-ax\geq\frac{b}{2}$.
	As $(1-a)\leq 2(1-b)$, we have $\big((1-b)-(1-a)x\big)\geq \frac{1-b}{2}$.
	Hence
	\[F\leq \left(\Div{a}{b}+\frac{(a-b)^2}{b(1-b)}\right)\cdot \m\leq \frac{3(a-b)^2}{b(1-b)}\cdot \m,\]
	for which we apply Lemma~\ref{lem-L4} to bound the divergence term.
\end{proof}

\begin{theorem}\label{thm-first-tail}
	Let $X\sim \Hyper(\hM,\hK, \m)$. 
	Let $0\leq\ell\leq \m$ be a real number with $\ell < \hK$ and $\m-\ell < \hM-\hK$. 
	Put $a = \frac{\ell}{\m}$, $b = \frac{\hK}{\hM}$ and $x = \frac{\m}{\hM}$. If $a\leq \frac{8}{5}b$, $(1-a)\leq 2(1-b)$ , $x\leq \frac{1}{4}$ and $\m a (1-a) \geq 4$, then we have
	\[Pr[X\geq \ell]\geq \sqrt{\frac{\pi}{320}}\cdot e^{-24}\cdot e^{-\frac{6(a-b)^2}{b(1-b)}\m}.\]
\end{theorem}
\begin{proof}
	Let $0\leq t\leq\sqrt{\m a(1-a)}$ be a real number such that $\ell+t$ is an integer and put $a'=\frac{\ell+t}{\m}$.
	As $\m a (1-a) \geq 4$,
	we have $t\leq\sqrt{\m a(1-a)}\leq\frac{\m a(1-a)}{4}$, so that
	\[a'=a+\frac{t}{\m}\leq a+\frac{a(1-a)}{4}\leq \frac{5}{4}a\leq 2b, \]
	and $1\leq a'\leq 1-a\leq 2(1-b)$. 
	We may hence apply Corollary~\ref{cor-C2} 
	and get
	\begin{align*}
	\Pr[X=\ell+t] & \geq \sqrt{\frac{\pi}{64\big(\m a'(1-a')\big)}}\cdot e^{-3\frac{(a+\frac{t}{\m}-b)^2}{b(1-b)}\m}\\
	& \geq \sqrt{\frac{\pi}{80\m a(1-a)}}\cdot e^{\frac{-3(a+\frac{t}{\m}-b)^2}{b(1-b)}\m}
	\end{align*}
	where we used that
	\[a'(1-a') \leq  \frac{5}{4} a (1-a).\]
	Since $(a+\frac{t}{\m}-b)^2\leq 2(a-b)^2+2(\frac{t}{\m})^2$, we have
	\[\frac{3(a+\frac{t}{\m}-b)^2}{b(1-b)}\m\leq \frac{6(a-b)^2}{b(1-b)}\m+\frac{6t^2}{\m b(1-b)}\]
	where
	\[\frac{6t^2}{\m b(1-b)}\leq \frac{6\m a(1-a)}{\m b(1-b)}=6\frac{a(1-a)}{b(1-b)}\leq 24.\]
	Hence we have
	\[\Pr[X =\ell+t]\geq\sqrt{\frac{\pi}{80\big(\m a(1-a)-t\big)}}\cdot e^{-\frac{6(a-b)^2}{b(1-b)}\m}\cdot e^{-24}.\]
	There are at least $\sqrt{\m a(1-a)}-1$ possible values of $t$.
	As $\m a (1-a) \geq 4$, we have
	\[\sqrt{\m a(1-a)}-1 \geq \frac{\sqrt{\m a(1-a)}}{2}.\]
	Thus summing over all possible possible values of $t$ yields the statement.
\end{proof}

\subsection{A useful tool}\label{sub-tool}

This subsection aims to build up a tool (Theorem~\ref{thm-tool}) for proving Corollary~\ref{cor-C2} in Appendix~\ref{sub-first-tail}.
We first introduce an entropy bound (Lemma~\ref{lem-entrboun}),
and then use this entropy bound prove Theorem~\ref{thm-tool}, in which we also prove Lemma~\ref{lem-L4}.

\begin{lemma}[Entropy bound~\cite{Ash12}]
	\label{lem-entrboun}
	Let $\HH(x) = - x \ln(x) - (1-x) \ln(1-x)$ be the entropy function.
	Let $0 \leq k \leq n$ be an integer and put $\alpha = \frac{k}{n}$. Then
	$$
	\frac{e^{n H(\alpha)}}{\sqrt{8 n \alpha (1-\alpha)}} \leq \binom{n}{k} \leq \frac{e^{n \HH(\alpha)}}{\sqrt{2 \pi n \alpha (1-\alpha)}}.
	$$
\end{lemma}

\begin{theorem}\label{thm-tool}
	Let $X \sim \Hyper(\hM, \hK, \m)$. Let $0 \leq \ell \leq \m$ be an integer with
	$\ell < \hK$ and $\m-\ell < \hM-\hK$. Put $a = \frac{k}{l}$, $b = \frac{\hK}{\hM}$, and $x = \frac{\m}{\hM}$, then we have we have
	$$
	\Prpr{X = \ell} \geq \sqrt{\frac{\pi}{32}} \sqrt{\frac{1 - x}{\m a (1-a)}} e^{-F}
	$$
	for
	$$
	F = \pa{\Div{a}{b} + \frac{x}{1-x} \cdot \frac{(a-b)^2}{2 (b- a x)((1-b)-(1-a) x)}} \cdot \m
	$$
	where $\Div{a}{b}$
	is the Kullback-Leibler divergence (Definition~\ref{def-KLdiv}).
\end{theorem}
\begin{proof}
	By the definition of the hypergeometric distribution, we have
	\[
	\Prpr{X = \ell} = \frac{\binom{\hK}{\ell} \binom{\hM - \hK}{\m - \ell}}{\binom{\hK}{\m}} = \frac{\binom{b \hM}{a \m} \binom{(1-b) \hM}{(1-a) \m}}{\binom{\hM}{\m}}
	\]
	By the entropy bound (Lemma~\ref{lem-entrboun}), we have
	\begin{align*}
	\binom{b \hM}{a \m} &\geq \frac{e^{b \hM H(\frac{a}{b} x)}}{\sqrt{8 b \hM \frac{a}{b} x (1-\frac{a}{b} x)}}\\
	\binom{(1-b) \hM}{(1-a) \m} &\geq \frac{e^{(1-b) \hM H(\frac{1-a}{1-b} x)}}{\sqrt{8 (1-b) \hM \frac{1-a}{1-b} x (1-\frac{1-a}{1-b} x)}}\\
	\binom{\hM}{\m} &\leq \frac{e^{\hM H(x)}}{\sqrt{2 \pi \hM x (1-x)}}
	\end{align*}
	Hence $\Prpr{X = \ell} \geq \frac{\sqrt{2 \pi}}{8} \frac{2^{\tilde{F}}}{\sqrt{G}}$ for
	\begin{align*}
	\tilde{F} &= \pa{b H(\frac{a}{b} x) + (1-b) H(\frac{1-a}{1-b} x) - H(x)} \cdot \hM\\
	G &= \frac{\m}{1-x} a (1-a) \pa{1-\frac{a}{b} x} \pa{1-\frac{1-a}{1-b} x}
	\end{align*}
	For $G$, we have
	\[
	G \leq \frac{\m}{1-x} a (1-a)
	\]
	For $\tilde{F}$, define
	$$
	\PP(a, b, x) = b \HH\pa{\frac{a}{b} x} + (1-b) \HH\pa{\frac{(1-a)}{(1-b)} x}- \HH(x)
	$$
	then we can write
	$$
	\tilde{F} = \PP(a, b, x) \cdot \hM
	$$
	Hence by Lemma~\ref{lem-claimA6}, we have
	\[
	\tilde{F} \geq \pa{-\Div{a}{b} x - \frac{(a-b)^2}{(1-x) (b - a x) ((1-b) - (1-a) x)} \cdot \frac{x^2}{2}} \hM
	\]
	which shows the result as $\tilde{F} = -x \cdot F$.
\end{proof}

\begin{lemma}\label{lem-claimA6}
	For $a, b, x$ defined as Theorem~\ref{thm-tool}, we have
	\label{clai:qfunlower}
	$$
	\PP(a, b, x) \geq - \Div{a}{b} \cdot x - \frac{{\pa{a-b}^2}}{(1-x) \cdot (b - a x) \cdot ((1-b) - (1-a) x)}\cdot \frac{x^2}{2}
	$$
\end{lemma}
\begin{proof}
	For fixed $a, b \in \pa{0, 1}$,
	let $C_{a, b} = \min\pa{\frac{b}{a}, \frac{1-b}{1-a}}$. Let
	\begin{align*}
	\PP(a, b, x) =\, & b \HH\pa{\frac{a}{b} x} + (1-b) \HH\pa{\frac{1-a}{1-b}x} - \HH(x)\\
	\QQ_{a, b}(x) =\, & -x a \ln\pa{\frac{a}{b}} - x (1-a) \ln\pa{\frac{1-a}{1-b}} + (1-x) \ln (1-x)\\
	&- b \cdot (1-\frac{a}{b} x) \ln \pa{1 - \frac{a}{b} x} - (1-b) (1-\frac{1-a}{1-b} x)\ln \pa{1 - \frac{1-a}{1-b} x} 
	\end{align*}
	Note that $\PP$ is defined for $x \in \pa{0, C_{a, b}}$
	and $\QQ$ is defined for $x \in \pa{-\infty, C_{a, b}}$.
	For $x \in \pa{0, C_{a, b}}$, we have $\PP(a, b, q) = \QQ_{a, b}(x)$.
	In other words, $\QQ$ is an extension of $\PP$ to non-positive values of $x$.
	
	A straight-forward computation shows that
	$\QQ_{a, b}$ is smooth on $\pa{-\infty, C_{a, b}}$ with
	\begin{align*}
	\QQ_{a, b}'(x) =\, & - a \ln\pa{\frac{a}{b}} - (1-a) \ln\pa{\frac{1-a}{1-b}} - \ln(1-x)\\
	&+ a \ln \pa{1 - \frac{a}{b} x} + (1-a) \ln\pa{1 - \frac{1-a}{1-b} x} \\
	\QQ_{a, b}''(x) =\, & - \frac{(a-b)^2}{(1-x)((1-b) - (1-a) x) (b - a x)}
	\end{align*}
	From these formulas, it is easy to see that
	\begin{align*}
	\QQ_{a, b}(0) &= 0\\
	\QQ_{a, b}'(0) &= -\Div{a}{b}
	\end{align*}
	and that $\QQ_{a, b}''(x)$ is non-decreasing in $x$.
	We use this to bound $\QQ$ by a second order tangent bound as follows:
	Since $\QQ_{a, b}''(x)$ is non-decreasing, we have
	$$
	\QQ_{a, b}''(w) \geq \QQ_{a, b}''(x) \quad \forall w \in \bra{0, x}
	$$
	Hence, by the fundamental theorem of calculus,
	for every $z \in \bra{0, x}$, we have
	\begin{align*}
	\QQ_{a, b}'(z) &= \QQ_{a, b}'(0) + \int_0^{z} \QQ_{a, b}''(w) \de{w}\\\
	&\geq \QQ_{a, b}'(0) + \int_0^{z} \QQ_{a, b}''(x) \de{w}\\
	&= \QQ_{a, b}'(0) + \QQ_{a, b}''(x) \cdot z
	\end{align*}
	Hence, again by the fundamental theorem of calculus,
	for every $y \in \bra{0, x}$, we have
	\begin{align*}
	\QQ_{a, b}(y) &= \QQ_{a, b}(0) + \int_0^{y} \QQ_{a, b}'(z) \de{z}\\
	&\geq \QQ_{a, b}(0) + \int_0^{y} \pa{\QQ_{a, b}'(0) + Q_{a, b}''(x) \cdot z} \de{z}\\
	&= \QQ_{a, b}(0) + Q_{a, b}'(0) \cdot y + \QQ_{a, b}''(x) \cdot \frac{y^2}{2}\\
	\end{align*}
	Setting $y = x$ and plugging in the formulas for the derivatives of $\QQ$ gives the result.
\end{proof}

\newpage
\bibliographystyle{plainurl}
\bibliography{bibliography}

\end{document}